\documentclass[a4paper, 12pt]{article}
\usepackage[margin=1in]{geometry}
\usepackage{amsmath, amssymb, amsfonts, amsthm}
\numberwithin{equation}{section}

\usepackage{enumitem}

\usepackage{pgf}
\usepackage{tikz}
\usetikzlibrary{arrows,automata}
\usetikzlibrary{positioning}
\usetikzlibrary{decorations.pathmorphing}
\usetikzlibrary{shapes.geometric}
\usetikzlibrary{fit}
\usetikzlibrary{backgrounds}

\newcommand{\mc}{\mathcal}
\newcommand{\mbb}{\mathbb}
\newcommand{\maxsub}{\mathrm{max}}

\DeclareMathOperator{\Tr}{Tr}

\DeclareMathOperator{\poly}{poly}

\newtheorem{lem}{Lemma}[section]
\newtheorem{thm}[lem]{Theorem}
\newtheorem{cor}[lem]{Corollary}
\newtheorem{prop}[lem]{Proposition}
\newtheorem{fact}[lem]{Fact}
\theoremstyle{definition}

\theoremstyle{remark}
\newtheorem{rem}[lem]{Remark}

\usepackage{hyperref}

\title{Mosaics of combinatorial designs for information-theoretic security}
\author{Moritz Wiese\footnote{Technical University of Munich, Chair of Theoretical Information Technology, Munich, Germany, and CASA: Cyber Security in the Age of Large-Scale Adversaries Exzellenzcluster, Ruhr Universit\"at Bochum, Bochum, Germany. Emails:\{wiese, boche\}@tum.de 
\newline\indent$^\dagger$Technical University of Munich, BMBF Research Hub 6G-life, Munich, Germany.
\newline\indent Part of this work has been presented at the IEEE International Symposium on Information Theory 2021. The conference version only treats the privacy amplification case and refers to the preprint of this paper at https://arxiv.org/abs/2102.00983v1 for details. It contains a slightly shortened introduction. Of Section \ref{sect:mosaics}, it only contains Subsection \ref{subsect:designs_defs} and examples $\mc M^{(1)}$ and $\mc M^{(4)}$ from Subsection \ref{subsect:mos_exs}. Of the results of Section 3, only Lemma \ref{lem:A_unif} and  Theorem \ref{thm:PA_muti_sec} are stated and proved. Section 5 is omitted completely.
}, Holger Boche$^{*\dagger}$}

\date{\today}

\begin{document}

\maketitle

\begin{abstract}
    We study security functions which can serve to establish semantic security for the two central problems of information-theoretic security: the wiretap channel, and privacy amplification for secret key generation. The security functions are functional forms of mosaics of combinatorial designs, more precisely, of group divisible designs and balanced incomplete block designs. Every member of a mosaic is associated with a unique color, and each color corresponds to a unique message or key value. Every block index of the mosaic corresponds to a public seed shared between the two trusted communicating parties. The seed set should be as small as possible. We give explicit examples which have an optimal or nearly optimal trade-off of seed length versus color (i.e., message or key) rate. We also derive bounds for the security performance of security functions given by functional forms of mosaics of designs.
\end{abstract}


\section{Introduction}

\subsection{Two problems of information-theoretic security}

A \textit{channel} $W:\mc X\to\mc Z$ is a stochastic matrix $W$ with rows indexed by the finite \textit{input alphabet} $\mc X$ and columns indexed by the finite \textit{output alphabet} $\mc Z$. The $(x,z)$ entry is nonnegative and denoted by $w(z\vert x)$. The sum of the entries of every row sums to 1, hence it defines a probability distribution on $\mc Z$. For the purpose of this paper, a \textit{wiretap channel} is determined by a single channel $W$. The interpretation is that a sender, Alice, wants to transmit a confidential message to a receiver, Bob, through a channel which accepts inputs from $\mc X$ and whose output is identical to the input, or whose error probability is as small as desired. An eavesdropper, Eve, obtains a noisy version of the input symbol $x\in\mc X$ through the channel $W$, in other words, she observes a random variable distributed according to $w(\,\cdot\,\vert x)$. The task now is to devise a \textit{security code} for the transmission of confidential messages which does not decrease the reliability of the channel to Bob, and which at the same time ensures that Eve learns nothing about the transmitted messages. In fact, we aim for \textit{semantic security}, by which we loosely mean that the security code should guarantee security no matter how the message is distributed on the message set. Two possible rigorous definitions of this concept will be given below. They guarantee \textit{unconditional security}, which means that no assumptions are made on Eve's computing power.

Another problem from information-theoretic security is \textit{privacy amplification}. Here, Alice and Bob share a random variable $X$ living on a finite set $\mc X$. Eve, the adversary, has access to a discrete random variable $Z$ correlated with $X$. The task is to apply a \textit{privacy amplification function} to $X$ such that the resulting random variable $A$ (the \textit{secret key} shared by Alice and Bob) is distributed approximately uniformly and such that Eve has no information about $A$. Again, the goal is to achieve semantic security. Although all distributions are fixed in this setting, it makes sense to require semantic security. For instance, it guarantees that even if Eve has the a priori knowledge that the key generated in the privacy amplification process has one of only two possible values, she is unable to tell which of these two is the one actually chosen. This property is sometimes called \textit{distinguishing security}, but it is well-known that it is equivalent to unconditional semantic security \cite{BTV}.

Practical scenarios will not in general translate directly into one of the two problems described above. In the wiretap scenario, the physical channel from Alice to Bob will generally be noisy as well, and an error-correcting code needs to be applied first to make the error probability on this channel as small as possible. In this case, the input alphabet $\mc X$ actually is the message set of the error-correcting code. Similarly, in secret key generation, two remote parties will not in general share a random variable $X$ from the outset. In order to establish such a random variable, an \textit{information reconciliation} protocol has to be performed using communication over a public channel. Eve obtains at least part of her correlated information $Z$ about $X$ as she observes the public messages exchanged during information reconciliation.

It follows that a security code or a privacy amplification function will generally be just one component of a modular scheme which as a whole ensures both ``reliability'' (viz.\ error-correction or information reconciliation) and semantic security as well as, in the privacy amplification setting, approximately uniform key distribution.

The two problems above are key techniques for the generation of information-theoretic security in communication and data storage systems. They can be building blocks for embedded security and security-by-design of such systems. An important feature of information-theoretic security is that it provides provable security even against attacks performed by a quantum computer. For this reason, the techniques developed here are of great importance for the development of future 6G mobile communication systems \cite{FettwBoche_personalTactile}. A first practical implementation is presented in \cite{Globecom21_modular}.

\subsection{Security functions}

Both for the wiretap and the privacy amplification scenario, we will assume that Alice and Bob can share an additional resource, a publicly known \textit{seed} $s$ chosen uniformly at random from the finite seed set $\mc S$. Then the basis both for security codes and privacy amplification functions are onto functions $f:\mc X\times\mc S\to\mc A$, where $\mc A$ is a finite set. We will call such a function a \textit{security function}. In the wiretap scenario, $\mc A$ will be the set of confidential messages; in privacy amplification, it represents the range of possible key values. In fact, in privacy amplification, $f$ is nothing else than the privacy amplification function, i.e., given a seed $s\in\mc S$ and a realization $x\in\mc X$ of the random variable $X$ shared by Alice and Bob, the secret key is chosen to be $f(x,s)$ (see Fig.~\ref{fig:PA_scenario}). For the wiretap channel, if Alice wants to send a confidential message $\alpha\in\mc A$ and shares the seed $s$ with Bob, she selects an element $x$ from the preimage $f_s^{-1}(\alpha)=\{x:f(x,s)=\alpha\}$ uniformly at random and transmits $x$. We call this process of selecting $x$ the \textit{randomized inverse} of $f$. By assumption, with high probability, or even with certainty, Bob receives the $x$ that was sent and decodes it into the original confidential message $\alpha=f(x,s)$, so the reliability of message transmission is preserved (see Fig.~\ref{fig:WT_scenario}).

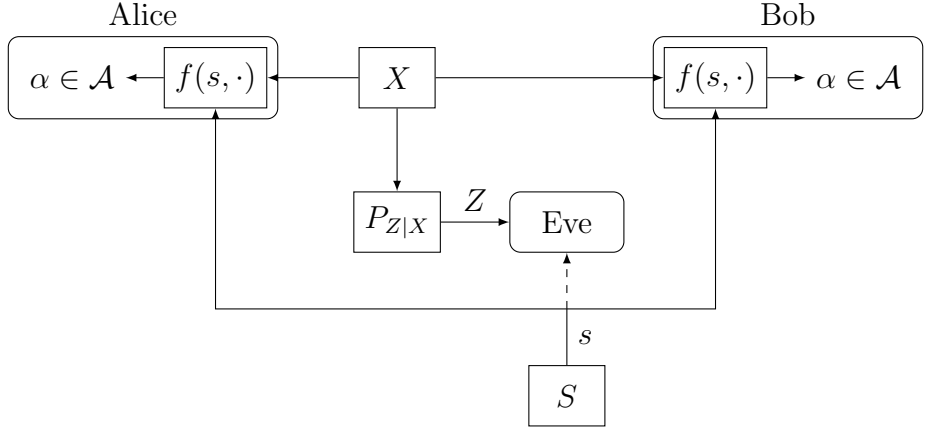
\begin{figure}
	\centering
	\begin{tikzpicture}[kasten/.style={draw, minimum height = .8cm, minimum width = 1cm, inner sep=4pt}, pfeil/.style={->, >=latex}
	]
	
    \node (key_Alice) {$\alpha\in\mc A$};
    \node[kasten, right = .5cm of key_Alice] (uhf_Alice) {$f(s,\cdot)$};
	\begin{scope}[on background layer]
        \node[draw, rounded corners, fit = (key_Alice) (uhf_Alice), label=Alice] {};
    \end{scope}
        
    \node[kasten, right = 1.2cm of uhf_Alice] (source) {$X$};
    
    \node[kasten, right = 3cm of source] (uhf_Bob) {$f(s,\cdot)$};
    \node[right = .5cm of uhf_Bob] (key_Bob) {$\alpha\in\mc A$};
    \begin{scope}[on background layer]
        \node[draw, rounded corners, fit = (uhf_Bob) (key_Bob), label=Bob] {};
    \end{scope}

    \node[kasten, below=1.1cm of source] (eve_channel) {$P_{Z\vert X}$};
    \node[draw, rounded corners, minimum height=.8cm, minimum width=1.5cm, right=.9cm of eve_channel] (eavesdropper) {Eve};
    
    \node[below = .6cm of eavesdropper] (seed_aux) {};
    \node[kasten, below = .6cm of seed_aux] (seed) {$S$};
    
    \draw[pfeil] (source) -> (uhf_Alice);
    \draw[pfeil] (source) -> (uhf_Bob);
    \draw[pfeil] (uhf_Alice) -> (key_Alice);
    \draw[pfeil] (uhf_Bob) -> (key_Bob);
    
    \draw[pfeil] (source) -> (eve_channel);
    \draw[pfeil] (eve_channel) -> (eavesdropper) node[midway, above]{$Z$};
    
    \draw[pfeil, dashed] (seed_aux.center) -- (seed_aux.center|-eavesdropper.south);
    \draw[pfeil] (seed_aux.center) -| (uhf_Alice);
    \draw[pfeil] (seed_aux.center) -| (uhf_Bob);
    \draw (seed) -- (seed_aux.center)  node[right, midway] {$s$};
 \end{tikzpicture}
	\caption{The privacy amplification scenario. The correlation between $X$ and Eve's observation $Z$ here is represented by the conditional probability $P_{Z\vert X}$. Alice and Bob are usually assumed to be connected by a public channel over which they can exchange messages. In particular, they can use this channel to share the seed. The arrow from $S$ to Eve is dashed because Eve may know the seed, but this is not necessary for the operability of the protocol.}\label{fig:PA_scenario}
\end{figure}

\begin{figure}
	\centering
	\begin{tikzpicture}[kasten/.style={draw, minimum height = .8cm, minimum width = 1cm, inner sep=4pt}, pfeil/.style={->, >=latex}
	]
	
    \node (message) {$\alpha\in\mc A$};
    \node[kasten, right = .5cm of message] (inv_uhf) {$f_s^{-1}$};
	\begin{scope}[on background layer]
        \node[draw, rounded corners, fit = (message) (inv_uhf), label=Alice] {};
    \end{scope}
        
    \node[right = 1cm of inv_uhf] (aux1) {};
    
    \node[kasten, right = 3cm of aux1] (uhf) {$f(s,\cdot)$};
    \node[right = .5cm of uhf] (rec_message) {$\alpha\in\mc A$};
    \begin{scope}[on background layer]
        \node[draw, rounded corners, fit = (uhf) (rec_message), label=Bob] {};
    \end{scope}

    \node[kasten, below=1.1cm of aux1] (eve_channel) {$W$};
    \node[draw, rounded corners, minimum height=.8cm, minimum width=1.5cm, right=.9cm of eve_channel] (eavesdropper) {Eve};
    
    \node[below = .6cm of eavesdropper] (seed_aux) {};
    \node[kasten, below = .6cm of seed_aux] (seed) {$S$};
    
    \draw[pfeil] (message) -> (inv_uhf);
    \draw[pfeil] (inv_uhf) -> (uhf) node[above, near start] {$x$};
    \draw[pfeil] (aux1.center) -> (eve_channel);
    \draw[pfeil] (uhf) -> (rec_message);
    \draw[pfeil] (eve_channel) -> (eavesdropper) node[midway, above]{$z$};
    
    \draw[pfeil, dashed] (seed_aux.center) -- (seed_aux.center|-eavesdropper.south);
    \draw[pfeil] (seed_aux.center) -| (inv_uhf);
    \draw[pfeil] (seed_aux.center) -| (uhf);
    \draw (seed) -- (seed_aux.center)  node[right, midway] {$s$};
 \end{tikzpicture}
	\caption{The wiretap scenario in the case where the channel between Alice and Bob is the identity channel. In principle, it is immaterial where the seed is generated. In practice, Alice will generate the seed and transmit it to Bob publicly. The arrow from $S$ to Eve is dashed because Eve may know the seed, but this is not necessary for the operability of the protocol.}\label{fig:WT_scenario}
\end{figure}
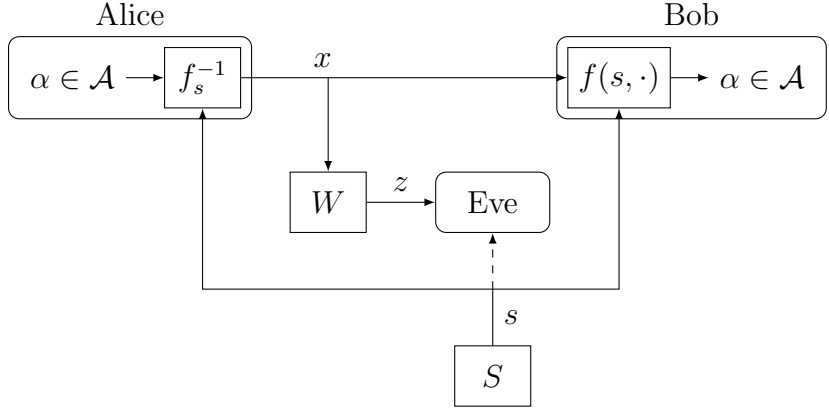

The \textit{color rate} of a security function $f:\mc X\times\mc S\to\mc A$, both in the wiretap and in the privacy amplification context, is given by\footnote{Throughout the paper, $\log$ will denote the logarithm to base 2. When we write $\exp(x)$, we mean $2^x$.}
\[
    \varrho=\frac{\log\lvert\mc A\rvert}{\log\lvert\mc X\rvert}
\]
(the name will be justified in the context of mosaics, see below). As $f$ is onto, this is a number between 0 and 1 which indicates the cost of establishing security as well as, in the privacy amplification scenario, approximately uniform key distribution. This is not a parameter to be optimized. Instead, given a required security level, the channel $W$ in the former and the joint probability $P_{XZ}$ in the latter situation determine a maximal possible color rate. The question is which $f$ achieve or come close to this rate. 

In the wiretap case, a common assumption is that Alice generates the seed, but then she has to use the unsecured channel to transmit it to Bob. This diminishes the overall communication rate significantly. The \textit{block rate}
\[
    \frac{\log\lvert\mc S\rvert}{\log\lvert\mc X\rvert}
\]
indicates how often the unsecured channel needs to be used for the transmission of the seed. It has been shown that in some scenarios the seed can be reused in order to make the loss of overall communication rate negligibly small asymptotically. Nevertheless, it is important to make the seed set $\mc S$ as small as possible. 

The use of a seed is not as problematic in the privacy amplification setting, since it is commonly assumed that there exists a public channel between Alice and Bob. For the purpose of seed sharing, it is sufficient that the public channel goes in one direction only. Usually, one still wants to keep the communication overhead on this public channel small, and this overhead can again be measured by the block rate.

Finally, we would like security codes and privacy amplification functions to be efficiently computable. For an underlying security function, this translates to the efficiency of computing $f(x,s)$ and the randomized inverse $f_s^{-1}(\alpha)$. A precise definition of what we mean by efficiency will be given below.

\subsection{Semantic security by mosaics of designs}

Semantic security can be seen as a per-message type of security. It means that the probability distribution of Eve's observations conditional on any message or key value should be indistinguishable from an arbitrary fixed distribution on Eve's observation space which is independent of the message or key distribution. This suggests to construct security functions $f:\mc X\times\mc S\to\mc A$ whose preimages $f^{-1}(\alpha)$ for every $\alpha\in\mc A$ have a structure suitable for  establishing this indistinguishability. 

Our goal in this paper is to systematically study security functions where every preimage $f^{-1}(\alpha)$ is the incidence relation of a balanced incomplete block design (BIBD) or a group divisible design (GDD) with point set $\mc X$ and block index set $\mc S$. Such a function defines a \textit{mosaic of designs} $(D_\alpha)_{\alpha\in\mc A}$, which is a family of designs on a common point set and a common block index set satisfying that every pair $(x,s)\in\mc X\times\mc S$ is incident in a unique $D_\alpha$. The security function corresponding to such a mosaic will be its \textit{functional form}. The precise definitions will be given in Section \ref{sect:mosaics}. 

Two aspects guide us in the construction of mosaics of designs: the trade-off of the color rate and the block rate, and the computational complexity of the functional form and its randomized inverse. We investigate the optimal trade-off of the color rate $\varrho$ vs.\ the block rate for functional forms of mosaics of BIBDs and GDDs. For mosaics of BIBDs with small color rate, the block rate can at best be equal to 1. In all other cases, the minimal block rate is approximately equal to $2\varrho$. In particular, if $\varrho<1/2$ and the mosaic consists of GDDs, then a block rate smaller than 1 is possible.

We construct families of examples which are close to optimal, or even optimal, in terms of this trade-off. Their color rates are distributed over the complete interval between 0 and 1, densely in all cases except for BIBDs of small color rate. Both for mosaics of BIBDs and GDDs, we need two different families in order to obtain sufficiently variable color rates $\varrho$. In both cases, at $\varrho\approx1/2$, the type of designs which is (close to) optimal in terms of this trade-off changes. 

To the best of our knowledge, we are the first to explicitly study semantic security for privacy amplification. In both scenarios, we measure the amount of semantic security offered by the functional form of a mosaic of BIBDs or GDDs using two alternative security metrics, one of them based on total variation distance, the other on Kullback-Leibler divergence. The upper bounds on these metrics rely on the local properties of the functional form, i.e., on the properties of BIBDs and GDDs. The wiretap channel and the distribution of the random variables $X$ and $Z$ only appear in these bounds through at most two R\'enyi entropies or divergences, which gives the bounds some robustness with respect to the knowledge about the channel or the random variables. 

We evaluate the security bounds in the most frequently studied scenarios of memoryless discrete or Gaussian wiretap channels and of privacy amplification for secret-key generation from discrete memoryless correlated sources. Unfortunately, block rate optimal mosaics of GDDs in the range where $\varrho$ is small only achieve a suboptimal security level in general. This is not due to our construction, but holds in general. Hence a block rate of at least 1 is necessary to achieve asymptotically perfect semantic security at the maximal message or key rate with mosaics of BIBDs or GDDs. For the other block rate optimal constructions, the bounds are asymptotically optimal in the benchmark scenarios. Additionally, in the case of privacy amplification, the regularity of BIBDs and GDDs immediately implies the perfect uniform distribution of the key generated by the functional form of a mosaic of designs.

All the mosaics we construct are explicit, by which we mean the efficient computability of the functional form and its inverse in the usual setting of asymptotic complexity. The examples are derived from well-known designs based on finite fields, so in some cases the explicitness is obvious. There is one case where some work is required to show explicitness.

\subsection{Related literature}\label{subsect:comparison}

Mosaics of combinatorial designs were introduced by Gnilke, Greferath and Pav\v cevi\'c \cite{GGP_mosaics}. Our method of constructing mosaics from resolvable designs or duals thereof is essentially due to them. The application of mosaics to construct functions with special desired properties is new, in particular, the analysis of color and block rates and of efficient computability of such functions. Mosaics generalize more specialized concepts like the tiling of a group with difference sets due to \'Custi\'c, Kr\v cadinac and Zhou \cite{CKZ_tilings}. A predecessor of what now is called mosaics was presented in \cite{Gref_Ther_CCC_designs} by Greferath and Therkelsen. General background on combinatorial designs can be found in the reference work of Beth, Jungnickel and Lenz \cite{BJL_book}.

The idea of separating privacy amplification from information reconciliation goes back to Bennett, Brassard and Robert \cite{BBR_PA} and Bennett, Brassard, Cr\'epeau and Maurer \cite{BBCM_PA}. Hayashi \cite{Hay_expdecr} extended the idea to the construction of security codes for the wiretap channel, where error correction is separated from the establishment of security. Like in \cite{BBR_PA}, \cite{BBCM_PA} and \cite{Hay_expdecr}, the weaker \textit{strong secrecy} criterion has been widely applied in information-theoretic security, where Eve's a priori knowledge is restricted to the true message or key distribution.

Semantic security ensures security no matter what the key or message distribution might be. Originating in complexity-based cryptography, it was adapted for (unconditional) information-theoretic security by Bellare, Tessaro and Vardy \cite{BTV} (the shorter, published version of which is \cite{BTV_published}). To the authors' knowledge, semantic security has only been considered for wiretap channels so far. In \cite{Hay_BB84}, Hayashi implicitly describes a technique for achieving semantic security for the quantum BB84 key distribution protocol.

\cite{BBR_PA}, \cite{BBCM_PA} and \cite{Hay_expdecr} used universal hash functions as security functions. Alternative choices in the privacy amplification scenario with strong secrecy are $\varepsilon$-almost dual universal hash functions (Hayashi \cite{Hay_almost_dual_UHF}) and strong randomness extractors (Maurer and Wolf \cite{MauW_extractors}). None of these choices guarantees perfect uniform distribution of the key. However, the seed required by randomness extractors can be very short. Seedless extractors have been used by Cheraghchi, Didier and Shokrollahi \cite{Cher} to ensure strong secrecy for the ``wiretap channel II'', where the eavesdropper may observe a fraction of his choosing of the transmitted codeword.

When applied as security functions in the wiretap scenario, it seems that the global property defining universal hash functions in general is not enough to ensure semantic security. Even with additional regularity properties (cf.~\cite{BT_poly_time, TV_hashing_published}), semantic security can only be shown for sufficiently symmetric channels. Usually, only strong secrecy is achievable.

Upper bounds on the semantic security metric for the wiretap channel which are comparable to ours were given by Hayashi and Matsumoto \cite[Lemma 21]{HayMat} and the authors \cite{BRI}, using security functions of a different type. The security functions of the former paper are defined in terms of group homomorphisms together with a regularity condition. The single efficiently computable example given in \cite[Remark 16]{HayMat} exhibits a block rate $\approx 2$, which is worse than for mosaics of designs with an optimal trade-off of block rate vs.\ color rate. The security functions of \cite{BRI} are induced by decompositions of complete biregular bipartite graphs into nearly Ramanujan graphs. A nonconstructive example of such a decomposition into Ramanujan graphs is given with a block rate of 1 independent of the color rate.

\subsection{Outline}

In Section \ref{sect:mosaics}, we define and analyze mosaics of BIBDs and GDDs. In Section \ref{sect:sec}, we define how we measure semantic security and give the bounds on the security metrics obtained from functional forms of mosaics of designs. These bounds are proved in Section \ref{sect:sec_proofs}. In Section \ref{sect:maxarc_proof}, we prove the explicitness of one of the examples of Section 2 for which this is not immediately obvious.

\section{Mosaics of combinatorial designs}\label{sect:mosaics}

\subsection{Definitions}\label{subsect:designs_defs}

Let $\mc X$ and $\mc S$ be finite sets. An \textit{incidence structure} $D=(\mc X,\mc S,I)$ on $(\mc X,\mc S)$ is determined by the \textit{incidence relation} $I$ on $\mc X\times\mc S$. An incidence structure $(\mc X,\mc S,I)$ is called \textit{empty} if $I=\emptyset$. If $x\,I\,s$, then $x$ and $s$ are called \textit{incident}. The \textit{incidence matrix} of an incidence structure $D=(\mc X,\mc S,I)$ is the 01-matrix $N$ with rows indexed by $\mc X$ and columns indexed by $\mc S$ such that $N(x,s)=1$ if and only if $x$ and $s$ are incident in $D$. 

A \textit{mosaic of incidence structures} on $(\mc X,\mc S)$ is a family $M=(D_\alpha)_{\alpha\in\mc A}$ of nonempty incidence structures on $(\mc X,\mc S)$ such that for every pair $(x,s)$ there exists a unique incidence structure $D_\alpha$ in which $x$ and $s$ are incident. We call $\mc A$ the \textit{color set} of $M$. Every $D_\alpha$ is called a \textit{member} of $M$. If $N_\alpha$ is the incidence matrix of $D_\alpha$, then $\sum_{\alpha\in\mc A}N_\alpha=J$, the all-ones matrix of appropriate size.

Any function $f:\mc X\times\mc S\to\mc A$ induces a mosaic $(D_\alpha)_{\alpha\in\mc A}$ of incidence structures, where $x$ and $s$ are incident in $D_\alpha$ if and only if $f(x,s)=\alpha$. We say that $f$ is the \textit{functional form} of this mosaic. Clearly, every mosaic $(D_\alpha)_{\alpha\in\mc A}$ on $(\mc X,\mc S)$ has a functional form $f:\mc X\times\mc S\to\mc A$.

We consider the case where every $D_\alpha$ is a combinatorial design. In the context of designs, we will call $\mc X$ the \textit{point set} and $\mc S$ the \textit{block index set}. We set $v=\lvert\mc X\rvert$ and $b=\lvert\mc S\rvert$. A $(v,k,r)$ \textit{tactical configuration} on $(\mc X,\mc S)$ is an incidence structure where every point $x$ is incident with precisely $r$ block indices and every block index $s$ is incident with precisely $k$ points. It holds that
\begin{equation}\label{eq:bireg}
    bk=vr.
\end{equation}

A $(v,k,\lambda)$ \textit{balanced incomplete block design (BIBD)} on $(\mc X,\mc S)$ is an incidence structure on $(\mc X,\mc S)$ such that every $s$ is incident with precisely $k$ points from $\mc X$ and such that any two distinct points from $\mc X$ are incident with precisely $\lambda\geq 1$ common block indices. Every $(v,k,\lambda)$ BIBD is a $(v,k,r)$ tactical configuration, where 
\begin{equation}\label{eq:BIBD_param_rel}
    r(k-1)=\lambda(v-1).
\end{equation}
The key equality when we want to establish security using a security function which is the functional form of a mosaic of BIBDs is that the incidence matrix $N$ of a  $(v,k,\lambda)$ BIBD satisfies
\begin{equation}\label{eq:BIBD_inc_matr}
    NN^T=(r-\lambda)I+\lambda J
\end{equation}
(here $I$ is the identity matrix of appropriate dimensions).

The second type of designs we consider are \textit{group divisible designs (GDDs)}. A $(u,m,k,\lambda_1,\lambda_2)$ GDD is based on a partition of $\mc X$ into $m$ \textit{point classes} of size $u$ each, so $v=um$. Every block index is incident with precisely $k$ points, and two points are incident with $\lambda_1\geq 0$ common block indices if they are contained in the same point class and with $\lambda_2\geq 1$ block indices otherwise. A $(u,m,k,\lambda_1,\lambda_2)$ GDD is a $(v,k,r)$ tactical configuration for $r$ satisfying
\begin{equation}\label{eq:GDD_param_rel}
    r(k-1)=\lambda_1(u-1)+\lambda_2(m-1)u.
\end{equation}
An equality similar to \eqref{eq:BIBD_inc_matr} holds for the incidence matrix $N$ of a GDD. Let $C$ be the 01-matrix with rows and columns indexed by $\mc X$ which has a $1$ in the $(x,x')$ entry if and only if $x$ and $x'$ are contained in the same point class. With a suitable ordering of the elements of $\mc X$, this is a block diagonal matrix with $m$ all-ones matrices of size $u$ each on the diagonal. Then
\begin{equation}\label{eq:GDD_inc_matr}
    NN^T=(r-\lambda_1)I+(\lambda_1-\lambda_2)C+\lambda_2J.
\end{equation}

For a BIBD or GDD $(\mc X,\mc S,I)$, the sets of the form $\{x:x\,I\,s\}$, where $s\in\mc S$, are usually called \textit{blocks} and the set $\mc S$ is identified with the multiset of blocks of the design. Occasionally, we will also speak of blocks and call the parameter $k$ the \textit{block size}. However, we will not identify $\mc S$ with a block multiset since we operate with multiple designs simultaneously. Hence the more cumbersome term ``block index set''. 

All mosaics in this paper will consist of tactical configurations with the same parameters $(v,k,r)$. Given a mosaic $(D_\alpha)_{\alpha\in\mc A}$, we will use the letter $a$ to indicate the cardinality of its color set $\mc A$. If $(D_\alpha)_{\alpha\in\mc A}$ is a mosaic of $(v,k,r)$ tactical configurations, then
\[
    a=\frac{v}{k}.
\]

In fact, the examples of mosaics constructed in the present paper will exclusively consist of BIBDs only or of GDDs only. BIBDs and GDDs together allow us to construct security functions with a wide range of color rates between 0 and 1. For a mosaic of BIBDs with constant block size $k$, note that $\lambda$ also has to be constant due to \eqref{eq:bireg} and \eqref{eq:BIBD_param_rel}.

\subsection{Some properties and examples of designs}\label{subsect:prop_ex_des}

If $D=(\mc X,\mc S,I)$ is an incidence structure, then its \textit{dual} is the incidence structure $D^T=(\mc S,\mc X,I^T)$ where $s\,I^T\,x$ if and only if $x\,I\,s$. Obviously, the incidence matrix of $D^T$ is the transpose of the incidence matrix of $D$. If $(D_\alpha)_{\alpha\in\mc A}$ is a mosaic of designs, then so is $(D_\alpha^T)_{\alpha\in\mc A}$.

A $(v,k,r)$ tactical decomposition $(\mc X,\mc S,I)$ is called \textit{resolvable} if the block index set $\mc S$ can be partitioned into subsets $\mc S_1,\ldots,\mc S_r$ such that for every $j\in\{1,\ldots,r\}$, every $x\in\mc X$ is incident with a unique $s\in\mc S_j$. (It is clear that such a partition necessarily has to have precisely $r$ elements.) Every $\mc S_j$ is called a \textit{parallel class} and contains $v/k$ block indices, in particular, $k$ divides $v$. 

The \textit{sum} of a mosaic is the incidence structure on $(\mc X,a\mc S)$, where $a\mc S$ is the disjoint union of $a$ copies of $\mc S$, and where a point $x$ is incident with the $\alpha$-th copy of $s\in\mc S$ if $x$ and $s$ are incident in $D_\alpha$. Note that the sum of a mosaic of tactical configurations is resolvable.

Two incidence structures $(\mc X,\mc S,I)$ and $(\mc X',\mc S',I')$ are called \textit{isomorphic} if there exist bijective mappings $\Phi_{\mc X}:\mc X\to\mc X'$ and $\Phi_{\mc S}:\mc S\to\mc S'$ such that $x\,I\,s$ if and only if $\Phi_{\mc X}(x)\,I'\,\Phi_{\mc S}(s)$. We also define that two mosaics $(D_\alpha)_{\alpha\in\mc A}$ on $(\mc X,\mc S)$ and $(D'_{\alpha'})_{\alpha'\in\mc A'}$ on $(\mc X',\mc S')$ are isomorphic if there exist bijective mappings $\Phi_{\mc X}:\mc X\to\mc X'$ and $\Phi_{\mc S}:\mc S\to\mc S'$ and $\Phi_{\mc A}:\mc A\to\mc A'$ such that $x\in\mc X$ and $s\in\mc S$ are incident in $D_\alpha$ for $\alpha\in\mc A$ if and only if $\Phi_{\mc X}(x)$ and $\Phi_{\mc S}(s)$ are incident in $D_{\Phi_{\mc A}(\alpha)}$.

A BIBD is called \textit{affine} or \textit{affine resolvable} if it is resolvable and if there exists a number $\mu>0$ such that any two distinct non-parallel blocks have precisely $\mu$ points in common. An \textit{affine plane} is an affine BIBD with $\mu=1$ and block size at least 2. Affine BIBDs have the property that their number of blocks is minimal among all resolvable BIBDs with the same number of points and parallel classes. This is a consequence of Bose's inequality, which states that
\begin{equation}\label{eq:Bose}
    b\geq v+r-1
\end{equation}
for resolvable BIBDs \cite[Corollary 8.6]{BJL_book}, and that equality holds if and only if the BIBD is affine.

Here we give the classical examples of affine designs, on which our constructions below will be based. This is no restriction, since all known affine BIBDs have the same parameters as the affine-geometric ones below or are Hadamard designs \cite[p. 128]{BJL_book}. We ignore the latter since they are limited to $v/k=a=2$, which only allows a very small color rate which vanishes asymptotically as $v$ increases.

Let $q$ be a prime power and $t\geq 2$. The $(q^t,q^{t-1},q^{t-2})$ BIBD $AG_{t-1}(t,q)$ has as block set the vector space $\mbb F_q^t$, the blocks are given by the hyperplanes of this vector space, i.e., all cosets of all $(t-1)$-dimensional subspaces, and the incidence relation is $\in$. These designs are affine resolvable, the parallel classes are given by the sets of nonintersecting hyperplanes. In the case $t=2$, one obtains the affine plane $AG(2,q)$, where the hyperplanes are called \textit{lines}.

\subsection{Block rate optimality}\label{subsect:bro}

We characterize block rate optimality for mosaics of BIBDs and GDDs.

\begin{lem}\label{lem:BRO_BIBD}
    Let $(D_\alpha)_{\alpha\in\mc A}$ be a mosaic of $(v,k,\lambda)$ BIBDs with $a\geq 2$ and color rate $\varrho$. Then
    \begin{equation}\label{eq:b_lb_BIBD}
        b\geq\max\left\{\frac{(v-1)ka^2}{v(k-1)},v\right\}.
    \end{equation}
    Setting
    \[
        \varrho_0(v,k)= 1-\frac{\log(v-1)+\log k-\log(k-1)}{2\log v},
    \]
    then this means for the block rate that
    \begin{equation}\label{eq:rho_lb_BIBD}
        \frac{\log b}{\log v}
        \begin{cases}
            >2\varrho & \text{if }\varrho>\varrho_0(v,k),\\
            \geq 1 & \text{if }\varrho\leq\varrho_0(v,k).
        \end{cases}
    \end{equation}
    If $\varrho\geq\varrho_0(v,k)$, then equality holds in \eqref{eq:b_lb_BIBD} if and only if $\lambda=1$. If $\varrho<\varrho_0(v,k)$, then equality holds in \eqref{eq:b_lb_BIBD} and \eqref{eq:rho_lb_BIBD} if and only if $b=v$. We call a mosaic of $(v,k,\lambda)$ BIBDs satisfying equality in one of these two cases 
    \emph{block rate optimal}. 
    \end{lem}

\begin{proof}
    Using \eqref{eq:bireg} and \eqref{eq:BIBD_param_rel},
    \[
        b=ra=\frac{\lambda(v-1)a}{k-1}\geq\frac{(v-1)ka^2}{v(k-1)}.
    \]
    Clearly, equality holds if and only if $\lambda=1$. The well-known Fisher's inequality \cite[Theorem II.2.6]{BJL_book} for BIBDs states $b\geq v$ if $k<v$, which settles \eqref{eq:b_lb_BIBD}.
    
    For the proof of \eqref{eq:rho_lb_BIBD}, observe that since $a=v^\varrho$, the maximum in \eqref{eq:b_lb_BIBD} is $v$ if and only if $\varrho\leq\varrho_0(v,k)$. If $\varrho>\varrho_0(v,k)$, then strict inequality has to hold due to $v>k$.
\end{proof}

We note that $\varrho_0(v,k)$ quickly approaches $1/2$ from below as $v$ increases.

We also consider the block rate for GDDs. This is connected to some subclasses of GDDs. First, we recall the classification of GDDs due to Bose and Connor \cite{BC_comb_GDDs}. A GDD is called
\begin{enumerate}
    \item \textit{singular} if $r=\lambda_1$,
    \item \textit{semi-regular} if $r>\lambda_1$ and $rk=v\lambda_2$,
    \item \textit{regular} if $r>\lambda_1$ and $rk>v\lambda_2$.
\end{enumerate}
Every GDD falls under exactly one of these categories. 

An important subclass of the semi-regular GDDs are the \textit{transversal designs}, which satisfy that every block intersects every point class in precisely one point. In this case $m=k$ and $\lambda_1=0$. We call a transversal design with these parameters a $(u,k,\lambda)$ TD, where $\lambda=\lambda_2$. Hanani \cite{Hanani_TDs} has shown that a $(u,k,\lambda)$ TD necessarily satisfies
\begin{equation}\label{eq:Hanani}
    k\leq\frac{\lambda u^2-1}{u-1}.
\end{equation}

\begin{lem}\label{lem:TD_BRO}
    Let $(D_\alpha)_{\alpha\in\mc A}$ be a mosaic of GDDs of constant block size $k$ and of color rate $\varrho$. Then
    \[
        \frac{\log b}{\log v}\geq2\varrho.
    \]
    Equality holds if and only if every $D_\alpha$ is an $(a,k,1)$ TD. We call such a mosaic \emph{block rate optimal}.
\end{lem}

\begin{proof}
    The parameters $v,b,k,r$ are the same for all members of the mosaic. Choose any $\alpha\in\mc A$ and assume that $D_\alpha$ is a $(u_\alpha,m_\alpha,k,\lambda_{1,\alpha},\lambda_{2,\alpha})$ GDD. By \eqref{eq:GDD_param_rel}
    \begin{align*}
        b
        &=ra\\
        &=\frac{\lambda_{1,\alpha}(u_\alpha-1)+\lambda_{2,\alpha}u_\alpha(m_\alpha-1)}{k-1}a\\
        &\geq\frac{\lambda_{2,\alpha}(v-u_\alpha)}{k-1}a.
    \end{align*}
    Equality here implies $\lambda_{1,\alpha}=0$, whence also $m_\alpha\geq k$. In this case,
    \[
        \frac{\lambda_{2,\alpha}(v-u_\alpha)}{k-1}
        =\frac{\lambda_{2,\alpha}ak}{k-1}\left(1-\frac{1}{m_\alpha}\right)
        =\frac{\lambda_{2,\alpha}ak(m_\alpha-1)}{m_\alpha(k-1)}
        \geq\lambda_{2,\alpha}a
        \geq a.
    \]
    Equality holds for $m_\alpha=k$ and $\lambda_{2,\alpha}=1$. Thus altogether we obtain
    \[
        b\geq a^2,
    \]
    with equality as claimed in the statement.
\end{proof}

Unfortunately, the block rate of any mosaic one of whose members is a semi-regular GDD cannot be much smaller than 1. This implies that the minimal possible color rate of a block rate optimal GDD quickly approaches $1/2$ from below as the number of colors increases.

\begin{lem}
	Consider a mosaic $M$ of $(v,k,r)$ tactical configurations and of color rate $\varrho<1/2$. Assume that its member $D_\alpha$ is a $(u,m,k,\lambda_1,\lambda_2)$ semi-regular GDD. If $\log a-\log(a-1)\leq\varepsilon$, then
	\[
	\frac{\log b}{\varrho\log v}\geq\frac{1}{\varrho}-\frac{\varepsilon}{\log a}.
	\]
\end{lem}

\begin{proof}
	From \eqref{eq:bireg} and the semi-regularity of $D_\alpha$, it follows that
	\begin{equation}\label{eq:b_semireg}
		b=a^2\frac{kr}{v}=a^2\lambda_2.
	\end{equation}
	Since $\log a=\varrho\log v$, this means that
	\begin{equation}\label{eq:GDD_rate}
		\frac{\log b}{\varrho\log v}=2+\frac{\log\lambda_2}{\log a}.
	\end{equation}
	
	The color rate is connected to $\lambda$ as follows. It was shown for semi-regular GDDs in \cite{BC_comb_GDDs} that 
	\begin{equation}\label{eq:Hanani_gen}
		b-1\geq v-m=m(u-1)
	\end{equation}
	(this generalizes Hanani's inequality) and that $m$ divides $k$, say $k=cm$. This implies $u=ac$, since $um=v=ak=acm$. Inserting this and \eqref{eq:b_semireg} in \eqref{eq:Hanani_gen}, one obtains
	\begin{equation}\label{eq:colorrate_lb}
		\varrho
		=\frac{\log a}{\log a+\log k}
		\geq\frac{\log a}{\log a+\log(a^2\lambda_2-1)-\log(ac-1)},
	\end{equation}
	hence
	\[
		\log(a^2\lambda_2-1)\geq\frac{\log a}{\varrho}-\log a+\log(a-1)
	\]
	and 
	\[
		\log\lambda_2\geq\left(\frac{1}{\varrho}-2\right)\log a-\varepsilon.
	\]
	Inserting this in \eqref{eq:GDD_rate} gives the result.
\end{proof}

\begin{cor}\label{cor:BRO_TD_neccond}
	A necessary condition for a mosaic $(D_\alpha)_{\alpha\in\mc A}$ of $(u,k,1)$ TDs to be block rate optimal is that the color rate $\varrho$ satisfies
	\[
	\varrho\geq\frac{\log u}{\log u+\log(u+1)}.
	\]
	Equality is attained if and only if every $D_\alpha$ is the dual of an affine plane.
\end{cor}

\begin{proof}
    We know that that $\lambda=1$ for mosaics of block rate optimal TDs. Using this and $c=1$ in \eqref{eq:colorrate_lb}, which holds for arbitrary $\varrho$, gives the lower bound.
    
    Assume that equality holds, and so Hanani's inequality holds with equality for every $D_\alpha$. According to Neumaier \cite[Corollary 3.8]{Neumaier_t12}, equality holds in Hanani's inequality for a TD $D$ if and only if $D$ is the dual of an affine BIBD. Thus every $D_\alpha$ is the dual of an affine BIBD. Since every $D_\alpha$ is a $(u,k,1)$ TD, any two distinct blocks of its dual $D_\alpha^T$ intersect in at most one point, hence $D_\alpha^T$ is an affine plane.
\end{proof}

We have seen that we cannot come close to block rate optimality for rates well below $1/2$ using mosaics which contain at least one semi-regular GDD. The same holds for mosaics which have at least one regular GDD as a member, since regular GDDs satisfy $b\geq v$ by \cite{BC_comb_GDDs}, so
\[
\frac{\log b}{\varrho\log v}
\geq\frac{1}{\varrho}.
\]

For color rates smaller than those in Corollary \ref{cor:BRO_TD_neccond}, the solution is to use singular GDDs. (However, we will see that singular block rate optimal GDDs give suboptimal bounds for semantic security for sufficiently large point set.) Bose and Connor show in \cite{BC_comb_GDDs} that every singular GDD is obtained by the multiplication of the points of a BIBD. We apply the same construction in order to obtain a mosaic $M=(D_\alpha)_{\alpha\in\mc A}$ of singular GDDs from a mosaic $M^*=(D_\alpha^*)_{\alpha\in\mc A}$ of $(v^*,k^*,\lambda^*)$ BIBDs on $(\mc X^*,\mc S^*)$. For an arbitrary positive integer $u$, replace each point $x^*\in\mc X^*$ by a class of $u$ copies of $x^*$. This gives the point set $\mc X$ of $M$. The block index set does not change, we set $\mc S=\mc S^*$. A point $x$ and a block index $s$ are defined to be incident in $D_\alpha$ if $x$ is a copy of an $x^*$ which is incident with $s$ in $D_\alpha^*$. Every $D_\alpha$ has parameters
\[
    u,\quad m=v^*,\quad k=uk^*,\quad \lambda_1=r=r^*,\quad \lambda_2=\lambda^*.
\]
Note that all members of $M$ have the same point class partition. We call $M$ the \textit{$u$-fold point multiple} of $M^*$. 

Conversely, one shows in the same way as in \cite{BC_comb_GDDs} that every mosaic of singular GDDs with the same parameters and the same point class partition is the $u$-fold point multiple of a mosaic of BIBDs.

\begin{lem}\label{lem:sing_GDDs}
    Let $M^*=(D^*_\alpha)_{\alpha\in\mc A}$ be a mosaic of $(v^*,k^*,\lambda^*)$ BIBDs with color rate $\varrho^*$. For a positive integer $u$, let $M$ be the $u$-fold point multiple of $M^*$. Then  $M$ has the color rate
    \[
        \varrho=\frac{\varrho^*\log v^*}{\log v^*+\log u},
    \]
    and satisfies
    \[
        \frac{\log b}{\varrho\log v}=\frac{\log b^*}{\varrho^*\log v^*}.
    \]

\end{lem}

\begin{proof}
    For the color rate, observe that $a=a^*$, and so
    \[
        \varrho=\frac{\log a}{\log v^*+\log u}=\frac{\log a^*}{\log v^*}\cdot\frac{\log v^*}{\log v^*+\log u},
    \]
    as claimed. The claim about the block rates follows from $v^\varrho=a=a^*=(v^*)^{\varrho^*}$ and $b=b^*$. 
\end{proof}

We see that if $M^*$ is close to block rate optimality and $\varrho^*\geq\varrho_0(v^*,k^*)$, then $M$ is close to block rate optimality as well. The color rates can be chosen arbitrarily small by choosing $u$ accordingly. A block rate optimal mosaic $M^*$ of BIBDs with color rate larger than $1/2$ will be constructed in Section \ref{subsect:mos_exs}.

\subsection{Complexity}\label{subsect:complexity}

With a view towards applications, we would like to be able to find examples of mosaics of designs whose functional form and randomized inverse are efficiently computable (in the Turing model of computation). By efficiency, we mean that it must be possible to do the computations in time polylogarithmic in $v$ and $b$. This is compatible with the usual requirements in coding theory, where encoding and decoding must be done in time polynomial in the blocklength. For this asymptotic definition to make sense, we will implicitly assume that the mosaic is part of an infinite family of mosaics where each color rate can be attained infinitely often and where $v$ is unbounded. This will be satisfied by all examples we give below.

Since $\mc X$ and $\mc S$ do not necessarily have a natural representation as a set of consecutive bit sequences, we define efficiency in terms of the functional form of an isomorphic mosaic defined on sets of integers. The choice of integers instead of bit strings allows us to ignore questions arising when the cardinality of a set is not a power of 2.

For any positive integer $n$, we write $[n]=\{1,\ldots,n\}$. We call the mosaic $M=(D_\alpha)_{\alpha\in\mc A}$ \textit{explicit} if there exists a mosaic $\tilde M=(\tilde D_j)_{j\in[a]}$ with point set $[v]$ and block index set $[b]$ which is isomorphic to $M$ and whose functional form $\tilde f:[v]\times[b]\to[a]$ satisfies
\begin{enumerate}[widest={(M2)}, leftmargin =*]
    \item[(M1)] $\tilde f(\tilde x,\tilde s)$ can be computed in time $\poly(\log v,\log b)$ (polynomial in $\log v$ and $\log b$) for all $\tilde x\in[v]$ and $\tilde s\in[b]$;
    \item[(M2)] there exists a mapping $g:[b]\times[a]\times[k]\to[v]$ such that $g(\tilde s,\tilde\alpha,\kappa)$ can be computed in time $\poly(\log b,\log v)$ for all $\tilde s,\tilde\alpha,\kappa$ and which for fixed $\tilde s\in[b]$ and $\tilde\alpha\in[a]$ is a bijection between $[k]$ and $\tilde f_{\tilde s}^{-1}(\tilde\alpha)$.
\end{enumerate}

\begin{rem}
Condition (M2) corresponds to the usual complexity-theoretic definition of \textit{strong explicitness} of graph families \cite{AB_comcompl}. Condition (M1) means the efficient distinction between different graphs, which is only of concern in the context of mosaics.
\end{rem}

If, in the wiretap channel case, Alice chooses the seed, which is the most likely scenario, then the order of the choice of seed $s$ and channel input $x$ can be reversed. So far, we have assumed that $s$ is chosen first and $x$ is chosen from $f_s^{-1}(\alpha)$. Due to \eqref{eq:bireg}, it is equivalent to first choose $x$ uniformly at random from $\mc X$ and then to choose $s$ from $f_x^{-1}(\alpha)=\{s\in\mc S:f(x,s)=\alpha\}$. We will see in one of the mosaics which we are going to construct that this can reduce the cost of computation. 

However, reversing the order of choosing $s$ and $x$ has a drawback. We already mentioned above that if the channel from Alice to Bob is used to first transmit the seed and then the confidential message, this incurs a loss of total communication rate, and that it is possible to make up for this by reusing the seed. Since $s$ depends on $x$ if the latter is chosen first, seed reuse is impossible in this case.

All functions constructed in this paper will be based on finite-field arithmetic. For real implementations, not all finite fields are equally suitable. However, in principle, the complexities are comparable. If $q=p^t$ for a prime $p$, then $\mbb F_q$ can be regarded as a vector space over $\mbb F_p$. If $\mbb F_q$ is represented in a polynomial basis, i.e., a basis of the form $\{1,\vartheta,\vartheta^2,\ldots,\vartheta^{t-1}\}$, then addition and subtraction in the field $\mbb F_q$ can be done in time $O(\log q)$. For multiplication and division, $O((\log q)^2)$ time is sufficient \cite{crypto_handbook}. A polynomial basis exists for all prime powers $q$ \cite{LN_finite_fields}.

\subsection{A general construction}\label{subsect:gen_constr}

We next present a method from which all examples of mosaics below will be constructed. A key ingredient for its construction are \textit{quasigroups.} A quasigroup on the finite set $\mc A$ is an array $L$ with entries from $\mc A$ and rows and columns indexed by $\mc A$ and which satisfies
\begin{enumerate}
    \item for every $\alpha,\gamma\in\mc A$ there is a unique $\beta\in\mc A$ such that $L(\alpha,\beta)=\gamma$,
    \item for every $\beta,\gamma\in\mc A$ there is a unique $\alpha\in\mc A$ such that $L(\alpha,\beta)=\gamma$.
\end{enumerate}
Every finite group is a quasigroup. If one labels the rows and columns of a quasigroup by a set which is not necessarily the same as $\mc A$, one obtains a \textit{Latin square}. Using quasigroups instead of Latin squares is more convenient in our setting.

A quasigroup $L$ on $\mc A$ and a quasigroup $\tilde L$ on $\tilde{\mc A}$ are called \textit{isomorphic} if there exists a bijective mapping $\Phi:\mc A\to\tilde{\mc A}$ such that $\tilde L(\Phi(\alpha),\Phi(\beta))=\Phi(\gamma)$ for all $\alpha,\beta,\gamma\in\mc A$.

The following theorem was already shown in \cite{GGP_mosaics} for the case of resolvable BIBDs, using Latin squares instead of quasigroups (which combinatorially amounts to the same thing). It is based on the idea that it should be possible to obtain a mosaic if one starts with a resolvable incidence structure, since the sum of a mosaic is resolvable.

\begin{thm}\label{thm:gen_constr}
    Let $D$ be a resolvable $(v,k,r)$ tactical configuration with incidence relation $I$. Let $\mc A$ be an index set for every parallel class of the blocks of $D$, and let $L$ be a quasigroup on $\mc A$. Then there exists a mosaic $M=(D_\alpha)_{\alpha\in\mc A}$ where each $D_\alpha$ is isomorphic to $D$, and there exists a mosaic $M^T=(D_\alpha^T)_{\alpha\in\mc A}$ where each $D_\alpha^T$ is isomorphic to $D^T$. 
    
    If $D$ is a GDD, then all $D_\alpha$ share their point class partitions with $D$. If $D^T$ is a GDD, then every $D_\alpha^T$ has the same point class partition as $D^T$ and $\lambda_1=0$.
\end{thm}

\begin{proof}
    The proof essentially is a reformulation of the proof of \cite{GGP_mosaics} together with the observation, already mentioned above, that if one has a mosaic and passes to the dual of every member of this mosaic, then one again obtains a mosaic. Our formulation of the proof will make it straightforward to derive the functional form of a mosaic constructed in this way.
    
    The block index set of $D$ can be written as $\mc R\times\mc A$, where $\mc R$ is an index set of cardinality $r$ for the parallel classes, and the blocks of each parallel class are labeled with a unique symbol from $\mc A$. Denote the point set of $D$ by $\mc P$. For every $p\in\mc P$ and $i\in\mc R$ there exists a unique $\alpha\in\mc A$ such that $p\,I\,(i,\alpha)$. We define the incidence structure $D_\alpha=(\mc P,\mc R\times\mc A,I_\alpha)$ by saying that $p\,I_\alpha\,(i,\beta)$ if and only if $p\,I\,(i,\gamma)$ for the unique $\gamma$ satisfying $L(\beta,\gamma)=\alpha$. This gives a mosaic. It follows directly from the construction and the quasigroup property of $L$ that all members of this mosaic are isomorphic to $D$. By dualization, one obtains a mosaic all members of which are isomorphic to $D^T$.
    
    It is clear that if $D$ is a GDD, then all $D_\alpha$ must have the same point class partition. For $D^T$, the point class partition corresponds to the partition of the blocks of $D$ into parallel classes, which is shared by all $D_\alpha$. This shows that all $D_\alpha^T$ have the same point class partition with $\lambda_1=0$.
\end{proof}

\begin{cor}\label{cor:funct_form}
    Assume the same conditions as in Theorem \ref{thm:gen_constr}. Let $\mc P$ be the point set of $D$ and $\mc R\times\mc A$ its block set, where $\mc R$ is an index set for the parallel classes and $\mc A$ is an index set for the elements of any parallel class. The functional form $f:\mc P\times(\mc R\times\mc A)\to\mc A$ of the mosaic $M$ constructed in Theorem \ref{thm:gen_constr} satisfies
    \[
        f(p;i,\beta)=L(\beta,\gamma)\quad\text{for the unique }\gamma\in\mc A\text{ with }p\,I\,(i,\gamma).
    \]
    The functional form $f^T:(\mc R\times\mc A)\times\mc P\to\mc A$ of $M^T$ satisfies $f^T(i,\beta;p)=f(p;i,\beta)$.
\end{cor}

The explicitness of a mosaic constructed as in Theorem \ref{thm:gen_constr} follows from the explicitness of the involved design $D$ and the quasigroup $L$. This is important in those cases where explicitness is not immediately clear from the functional form of the mosaic, like for the mosaics $\mc M^{(2)}$ of the next section.

We say that a quasigroup $L$ on $\mc A$ is \textit{explicit} if there exists an isomorphic quasigroup $\tilde L$ over $[a]$ such that
\begin{enumerate}[widest={(L2)}, leftmargin =*]
    \item[(L1)] $\tilde L(\tilde\beta,\tilde\gamma)$ can be computed in time $\poly(\log a)$ for all $\tilde\beta,\tilde\gamma\in[a]$,
    \item[(L2)] $\tilde L(\tilde\beta,\cdot)=\tilde\alpha$ can be solved in time $\poly(\log a)$ for all $\tilde\beta,\tilde\gamma\in[a]$.
\end{enumerate}

Let $D=(\mc X,\mc S,I)$ be a resolvable $(v,k,r)$ tactical configuration with $\mc S=\mc R\times\mc A$, where $\mc R$ is an index set for the parallel classes and $\mc A$ for the elements of each parallel class. We call $D$ \textit{explicit} if there exists an isomorphic resolvable tactical configuration $\tilde D=([v],[r]\times[a],\tilde I)$ satisfying
\begin{enumerate}[widest={(D2)}, leftmargin =*]
    \item[(D1)] for every $\tilde x\in[v]$ and $\tilde\imath\in[r]$, the unique $\tilde\alpha\in[a]$ satisfying $\tilde x\,\tilde I\,(\tilde\imath,\tilde\alpha)$ can be computed in time $\poly(\log v,\log r)$;\label{crit:D1}
    \item[(D2)] there exists a mapping $g:[r]\times[a]\times[k]\to[v]$ whose values are computable in time $\poly(\log b,\log k)$ and which satisfies that $\tilde\kappa\mapsto g(\tilde\imath,\tilde\alpha,\tilde\kappa)$ is a bijection between $[k]$ and the set of points in $[v]$ incident in $\tilde D$ with the block index $(\tilde\imath,\tilde\alpha)$.
\end{enumerate}
We call the dual $D^T$ of $D$ \textit{explicit} if there exists a resolvable tactical configuration $\tilde D=([v],[r]\times[a],\tilde I)$ isomorphic to $D$ which satisfies (D1) and
\begin{enumerate}[widest={(D2)$^T$}, leftmargin=*]
	\item [(D2)$^T$] there exists a mapping $g^T:[k]\times[r]\times[a]\to[v]$ whose values are computable in time $\poly(\log b,\log k)$ and which satisfies that $(\tilde{\imath},\tilde{\alpha})\mapsto g^T(\tilde{\kappa},\tilde\imath,\tilde\alpha)$ is a bijection between $[r]\times[a]$ and the set of blocks in $[b]$ incident in $\tilde D$ with the point $\tilde\kappa$.
\end{enumerate}

\begin{thm}\label{thm:comp_compl}
    Assume the conditions as in Theorem \ref{thm:gen_constr}. The mosaic $M$ constructed in Theorem \ref{thm:gen_constr} is explicit if both $D$ and $L$ are explicit. Its dual $M^T$ is explicit if $D^T$ and $L$ are explicit.
\end{thm}

\begin{proof}
	Let $\tilde D$ be a design as in the definition of explicitness of $D$ and let $\tilde L$ a quasigroup as in the definition of explicitness of $L$. The design $\tilde M$ constructed from $\tilde D$ and $\tilde L$ is isomorphic to $M$.
    
    In order to check (M1), let $\tilde f:[v]\times[r]\times[a]\to[a]$ be the functional form of $\tilde M$. Choose any $\tilde x\in[v],\tilde\imath\in[r]$ and $\tilde\beta\in[a]$. Then by Corollary \ref{cor:funct_form}, $\tilde f(\tilde x;\tilde\imath,\tilde\beta)=\tilde L(\tilde\beta,\tilde\gamma)$ for the unique $\tilde\gamma$ satisfying $\tilde x\,\tilde I\,(\tilde\imath,\tilde\gamma)$. By (D1), this $\tilde\gamma$ can be found in time $\poly(\log v,\log r)$, and $\tilde L(\tilde\beta,\tilde\gamma)$ can be computed in time $\poly(\log a)$ by (L1). Thus $\tilde f(\tilde x;\tilde\imath,\tilde\beta)$ can be computed in time $\poly(\log v,\log b)$.
    
    In order to check (M2), fix any $(\tilde\imath,\tilde\beta)\in[r]\times[a]$ and $\tilde\alpha\in[a]$. By (L2), the $\tilde\gamma$ satisfying $\tilde L(\tilde\beta,\tilde\gamma)=\tilde\alpha$ can be found in time $\poly(\log a)$. By (D2), there exists a mapping $\kappa\mapsto g(\tilde\imath,\tilde\gamma,\kappa)$ which enumerates all points incident with $(\tilde\imath,\tilde\gamma)$ in $\tilde D$ and whose values can be computed in time $\poly(\log b,\log k)$. The set of these points equals $\tilde f_{(\tilde\imath,\tilde\beta)}^{-1}(\tilde\alpha)$.
    
    Altogether, this proves the explicitness of $M$. The explicitness of $M^T$ is shown similarly.
\end{proof}

\subsection{Examples of (nearly) block rate optimal mosaics}\label{subsect:mos_exs}

We present four families of mosaics. Not all of these are block rate optimal, but those which are not are arbitrarily close to optimality for sufficiently large point sets. There is a family for each combination of the cases
\begin{enumerate}
    \item color rate $\varrho\geq 1/2$ or $\varrho\leq 1/2$ (roughly),
    \item BIBD or GDD.
\end{enumerate}
The sets of color rates will be dense except for the case of BIBDs with small color rates. 

In all cases we will use Theorem \ref{thm:gen_constr}. Thus in every case the key is to find a single resolvable design with the desired parameters.

\paragraph{BIBD and $\varrho\leq 1/2$:}

For this case we build our construction on the affine designs. Fix an integer $t\geq 2$ and a prime power $q$ and let $v,k,\lambda$ etc.\ be the parameters of the BIBD $AG_{t-1}(t,q)$. Then
\[
    v=q^t,\quad b=\frac{q(q^t-1)}{q-1},\quad r=\frac{q^t-1}{q-1},\quad k=q^{t-1},\quad\lambda=q^{t-2},\quad a=q.
\]
Hence the color rate of the mosaic $M^{(1)}_{t,q}$ we obtain from $AG_{t-1}(t,q)$ with the construction of Theorem \ref{thm:gen_constr} is 
\[
    \varrho=\frac{1}{t}.
\]
We have $1/t>\varrho_0(v,k)$ only if $t=2$. In this case, $M^{(1)}_{t,q}$ is block rate optimal since $\lambda=1$. 

If $t\geq 3$, then $M^{(1)}_{t,q}$ could only be block rate optimal if it were square, which is not the case. However, since $AG_{t-1}(t,q)$ is affine, it is a consequence of Bose's inequality \eqref{eq:Bose} that the block rate of $M^{(1)}_{t,q}$ is minimal among those mosaics constructed from any of the known resolvable BIBDs with $v=q^t$ and color rate $1/t$. The block rate satisfies
\[
    \frac{\log b}{\log v}
    \leq1+\frac{1}{t}\left(1-\frac{\log(q-1)}{\log q}\right).
\]
Thus for fixed color rate $1/t$, one gets closer to block rate optimality by increasing $q$.

Every hyperplane of $AG_{t-1}(t,q)$ can be represented by a unique pair $(h,\alpha)$, where $\alpha\in\mbb F_q$ and $h$ is a nonzero element of $\mbb F_q^t$ whose first nonzero component is normalized to $1$. We denote the set of these $h$ by $\mc R$. The hyperplane corresponding to $(h,\alpha)$ is the set of points $x$ satisfying $h\cdot x=\alpha$, where $h\cdot x=\sum_ih_ix_i$. Different $h$ give different parallel classes and different $\alpha$ with a fixed $h$ indicate different parallel hyperplanes in the parallel class corresponding to $h$. 

The natural quasigroup to construct a mosaic from $AG_{t-1}(t,q)$ is the additive group of $\mbb F_q$. Then a point $x\in\mbb F_q^t$ and an element $(h,\beta)$ of the block index set are incident in $D_\alpha$ if and only if $x$ is incident with $(h,\alpha-\beta)$ in $AG_{t-1}(t,q)$. The functional form $f:\mbb F_q^t\times(\mc R\times\mbb F_q)\to\mbb F_q$ of $M^{(1)}_{t,q}$ is given by
\[
    f(x;h,\beta)=h\cdot x+\beta.
\]
This immediately shows that the family
\[
	\mc M^{(1)}=\{M^{(1)}_{t,q}:t\geq 2,q\text{ prime power}\}
\]
is explicit. 

\paragraph{BIBD and $\varrho\geq1/2$:} Fix a positive integer $t\geq 2$ and an integer $\ell$ between $1$ and $t$. For $q=2^t$, let $Q:\mbb F_q^2\to\mbb F_q$ be an irreducible quadratic form, i.e., a polynomial of the form
\[
    Q(x,y)=\eta_1x^2+\eta_2xy+\eta_3y^2
\]
which cannot be factored into linear forms. Such a quadratic form exists for all $q$. Choose an arbitrary subgroup $H$ of order $2^\ell$ of the additive group of $\mbb F_q$ and consider the set 
\[
    \mc X=\{(x,y):Q(x,y)\in H\}.
\]
It was proved by Denniston \cite{Denniston_maxarc} that $\mc X$ has
\begin{equation}\label{eq:Denniston_card_X}
    v=1+(2^t+1)(2^\ell-1)
\end{equation}
elements and that every line of $AG(2,q)$ has either $2^\ell$ or no points in common with $\mc X$. 

We will regard $\mc X$ as a subset of $AG(2,q)$. It is not hard to see \cite[Corollary VIII.5.21]{BJL_book} that if we denote by $\mc S$ the set of nontrivial intersections of lines of $AG(2,q)$ with $\mc X$, then $D=(\mc X,\mc S,\in)$ is a resolvable $(v,k,1)$ BIBD with $k=2^\ell$. Since $r=2^t+1$ by \eqref{eq:BIBD_param_rel}, the set of parallel classes of $D$ is in one-to-one relation with the set of parallel classes of lines in $AG(2,q)$. In fact, if $\ell=t$, then $D=AG(2,q)$.

Applying Theorem \ref{thm:gen_constr}, one constructs a mosaic $M^{(2)}_{t,\ell,H}$ with the parameters
\begin{align*}
	&v=2^\ell(2^t+1-2^{t-\ell}),\quad b=(2^t+1)(2^t+1-2^{t-\ell}),\\
	&r=2^t+1,\quad k=2^\ell,\quad\lambda=1,\quad a=2^t+1-2^{t-\ell},\\
	&\varrho=\frac{\log(2^t+1-2^{t-\ell})}{\ell+\log(2^t+1-2^{t-\ell})}\approx\frac{t}{t+\ell}.
\end{align*}
Since $\lambda=1$, the mosaic $M^{(2)}_{t,\ell,H}$ is block rate optimal and satisfies
\begin{equation}\label{eq:Denniston_BR_ratio}
	\frac{\log b}{\varrho\log v}=\frac{\log b}{\log a}
	=2+\left(\frac{\log(2^t+1)}{\log(2^t+1-2^{t-\ell})}-1\right).
\end{equation}

For every $t$ and $\ell$, it is possible to choose a subgroup $H_{t,\ell}$ such that the resulting family
\[
	\mc M^{(2)}=\{M^{(2)}_{t,\ell,H_{t,\ell}}:t\geq2,1\leq\ell\leq t\}
\]
is explicit. Some work has to be done in order to show this, which we postpone to Section \ref{sect:maxarc_proof}. Moreover, every number between 1/2 and 1 can be approximated arbitrarily closely by the color rates of elements of $\mc M^{(2)}$ for sufficiently large $t$ and $\ell$.

\paragraph{GDD and $\varrho<1/2$:} Fix a positive integer $t$ and a nonnegative $\ell$ between 0 and $t$. Denote the elements of the explicit family $\mc M^{(2)}$ constructed above by $M^{(2)}_{t,\ell}$ (we omit the subgroups here in order to simplify notation). Choose an integer $u$ and let $M^{(3)}_{t,\ell,u}$ be the $u$-fold point multiple of $M^{(2)}_{t,\ell}$. Its parameters are
\begin{align*}
    &u,\quad m=2^\ell(2^t+1-2^{t-\ell}),\quad b=(2^t+1)(2^t+1-2^{t-\ell}),\\
    &r=2^t+1,\quad k=2^\ell u,\quad\lambda_1=2^t+1,\quad\lambda_2=1,\quad a=2^t+1-2^{t-\ell}.
\end{align*}
By Lemma \ref{lem:sing_GDDs}, its color rate is
\[
    \varrho=\frac{\log(2^t+1-2^{t-\ell})}{\ell+\log(2^t+1-2^{t-\ell})+\log u}\approx\frac{t}{t+\ell+\log u}
\]
and the ratio of the block rate and the color rate is given by \eqref{eq:Denniston_BR_ratio}. The color rate is smaller than $1/2$ for sufficiently large $u$.

Denote the point set of $M^{(2)}_{t,\ell}$ by $\mc X^*$ and its block index set by $\mc S^*$. Let $f^*:\mc X^*\times\mc S^*\to\mc A^*$ be the functional form of $M^{(2)}_{t,\ell}$. The point set of $M^{(3)}_{t,\ell,u}$ can be taken to be $\mc X=\mc X^*\times[u]$, the block index set and the color set remain the same as for $D^*$, so $\mc S=\mc S^*$ and $\mc A=\mc A^*$. The functional form $f:\mc X\times\mc S\to\mc A$ of $M^{(3)}_{t,\ell,u}$ satisfies
\[
    f(x^*,i;s)=\alpha\quad\text{if and only if}\quad f^*(x^*,s)=\alpha
\]
for $x^*\in\mc X^*,i\in[u],s\in\mc S$ and $\alpha\in\mc A$. The explicitness of the family
\[
	\mc M^{(3)}=\{M^{(3)}_{t,\ell,u}:t\geq 2,1\leq\ell\leq t,u\geq 1\}
\]
follows from that of $\mc M^{(2)}$.

By the discussion in Section \ref{subsect:bro}, mosaics of singular GDDs give the best approximation to block rate optimality among mosaics of GDDs with a small color rate if the point set is sufficiently large. The ratio of the block and the color rates is given by \eqref{eq:Denniston_BR_ratio}. All numbers between 0 and 1 can be approximated arbitrarily well by the color rates of suitable members of $\mc M^{(3)}$.

\paragraph{GDD and $\varrho\geq 1/2$:} If one deletes some of the parallel classes from the block set of $AG(2,q)$, where $q$ is a prime power, then one obtains the dual of a transversal design. Assume we keep $k\geq 2$ of the parallel classes of $AG(2,q)$. Call the resulting design $D^T$ and set $D=(D^T)^T$. The point set $\mc X$ of $D$ consists of lines of $AG(2,q)$ and the block index set $\mc S$ of $D$ consists of all the points of $AG(2,q)$. Two points $x,x'\in\mc X$ are incident with a common block index $s$ if and only if they intersect as lines in $AG(2,q)$, and so parallel classes of $D^T$ translate into point classes of $D$. If $x,x'$ are not in the same point class of $D$, then in $D^T$, their corresponding lines intersect in a unique point. In $D$, this means that two points from different point classes are incident with a unique block index, and so $D$ is a $(q,k,1)$ TD. 

Letting $\mc R$ denote the set of remaining parallel classes of lines, we construct from this transversal design a mosaic $M^{(4)}_{k,q,\mc R}$ as in Theorem \ref{thm:gen_constr}, using the natural additive group structure of $\mbb F_q$ on every parallel class of $AG(2,q)$. We obtain a mosaic with 
\[
    u=q,\quad k,\quad b=q^2,\quad \lambda=1,\quad a=q.
\]
Thus $M^{(4)}_{k,q,\mc R}$ has color rate 
\[
    \varrho=\frac{\log q}{\log q+\log k}.
\]
Since $k$ ranges between 2 and $q+1$, $\varrho$ is a number between
\[
    \frac{\log q}{\log q+\log(q+1)}\quad\text{and}\quad\frac{\log q}{1+\log q}.
\]
The block rate is optimal by Lemma \ref{lem:TD_BRO}.

The point set $\mc X$  of $M^{(4)}_{k,q,\mc R}$ has the structure of a Cartesian product, $\mc X=\mc R\times\mbb F_q$. For the discussion of the functional form of the mosaic, we assume that $k\leq q$ and that $\mc R$ is given by a subset of $\mbb F_q$. Then $x=(c,d)\in\mc X$ corresponds to the line $\{(u,cu+d):u\in\mbb F_q\}$ in $AG(2,q)$. The case $k=q+1$ can be treated analogously and corresponds to a mosaic whose members all are isomorphic to the dual of $AG(2,q)$.

A point $x=(c,d)\in\mc X$ and a block $s=(s_1,s_2)\in\mc S=\mbb F_q^2$ are incident in $D$ if $cs_1+d=s_2$. They are incident in $D_\alpha$ if $cs_1+d-\alpha=s_2$, where $\alpha\in\mbb F_q$. Thus
\[
    f(x,s)=f(c,d;s_1,s_2)=s_2-cs_1+d.
\]
Given $\alpha\in\mbb F_q$ and $s=(s_1,s_2)\in\mbb F_q^2$, one can find those $x\in\mc X$ which are incident with $s$ by taking any $c\in\mc R$ and solving for $d=\alpha-s_2+cs_1$. In this way, one obtains the randomized inverse of $f$. This can be done efficiently if $\mc R$ can be enumerated efficiently. Clearly, such an $\mc R=\mc R_{k,q}$ exists for every $k$. This gives us an explicit family
\[
	\mc M^{(4)}=\{M^{(4)}_{k,q,\mc R_{k,q}}:2\leq k\leq q+1,q\text{ prime power}\}.
\]
All numbers between $1/2$ and 1 can be approximated arbitrarily well by the color rates of members of this family.

\paragraph{Discussion.}

All our examples are constructed using Theorem \ref{thm:gen_constr}, hence all members of these designs are either themselves resolvable or duals of resolvable designs. We do not know whether mosaics of BIBDs or GDDs with constant block size exist which are not resolvable or dually resolvable. Such a construction would be particularly relevant for cases where mosaics of resolvable designs cannot be block rate optimal. For instance, a block rate optimal mosaic of BIBDs with color rate smaller than $1/2$ must be square, and consequently cannot be resolvable.

It would also be desirable to construct a family of mosaics of BIBDs which is close to block rate optimality and whose color rates are dense in the interval between 0 and $1/2$.

\subsection{Related structures}

\subsubsection{Universal hash functions.} A function $f:\mc X\times\mc S\to\mc A$ is called a \textit{universal hash function} if for all distinct $x,x'\in\mc X$,
\begin{equation}\label{eq:UHF}
	\frac{\lvert\{s:f(x,s)=f(x',s)\}\rvert}{b}\leq\frac{1}{a}
\end{equation}
(where, as usual, $\lvert\mc X\rvert=v$, $\lvert\mc S\rvert=b$ and $\lvert\mc A\rvert=a$). The left-hand side of \eqref{eq:UHF} can be interpreted as the probability that the values assigned to $x$ and $x'$ by $f$ ``collide'' if the seed is chosen uniformly at random. Let $(D_\alpha)_{\alpha\in\mc A}$ be the mosaic of incidence structures induced by $f$ as described in Section \ref{subsect:designs_defs}. Stinson \cite{Stinson} has shown that the maximal collision probability of $f$ is minimal if the sum $D$ of $(D_\alpha)_{\alpha\in\mc A}$ is a BIBD (recall the definition of the sum of a mosaic in Section \ref{subsect:prop_ex_des}). 

\begin{lem}[\cite{Stinson}]
	Any onto function $f:\mc X\times\mc S\to\mc A$ satisfies
	\[
		\frac{\lvert\{s:f(x,s)=f(x',s)\}\rvert}{b}\geq\frac{v-a}{a(v-1)}
	\]
	for at least one pair of distinct points $x,x'\in\mc X$. Equality holds for all distinct $x,x'\in\mc X$ if and only if the sum $D$ of the mosaic of incidence structures $(D_\alpha)_{\alpha\in\mc A}$ induced by $f$ is a resolvable BIBD.
\end{lem}

A universal hash function $f$ for which the sum $D$ of the corresponding mosaic $(D_\alpha)_{\alpha\in\mc A}$ is a BIBD is called \textit{optimally universal}. It follows immediately that a mosaic $(D_\alpha)_{\alpha\in\mc A}$ of BIBDs with common parameters $(v,k,\lambda)$ gives rise to an optimally universal hash function, since all blocks have the same size, and for distinct $x,x'\in\mc X$
\[
	\lvert\{s:f(s,x)=f(s,x')\}\rvert=\sum_{\alpha\in\mc A}\lvert\{s:f(s,x)=f(s,x')=\alpha\}\rvert=a\lambda.
\]
This proves the first part of the following lemma.

\begin{lem}\label{lem:rel_to_UHF}
	Let $M=(D_\alpha)_{\alpha\in\mc A}$ be a mosaic of $(v,k,r)$ tactical configurations on $(\mc X,\mc S)$ with functional form $f:\mc X\times\mc S\to\mc A$.
	\begin{enumerate}
		\item If every $D_\alpha$ is a $(v,k,\lambda)$ BIBD, then $f$ is optimally universal.
		\item If $M$ consists of $(u,m,k,\lambda_1,\lambda_2)$ GDDs with a common point class partition, then 
		\begin{enumerate}
			\item if every $D_\alpha$ is either semi-regular, or singular with $a=1$, then $f$ is a universal hash function;
			\item if the $D_\alpha$ are singular with $a\geq 2$, then $f$ is not a universal hash function.
		\end{enumerate} 
	\end{enumerate}
\end{lem}

\begin{proof}
	It remains to prove the second part of the lemma. We analyze the parameters of the mosaic. For distinct points $x,x'$,
	\begin{align*}
		\lvert\{s:f(s,x)=f(s,x')\}\rvert
		&=\sum_{\alpha\in\mc A}\lvert\{s:f(s,x)=f(s,x')=\alpha\}\rvert\\
		&=\begin{cases}
			a\lambda_1 & \text{if }x,x'\text{ are contained in the same point class,}\\
			a\lambda_2 & \text{else.}
		\end{cases}
	\end{align*}
	Since $a/b=1/r$ and $a=v/k$, we have for $i=1,2$
	\[
		\frac{a\lambda_i}{b}\leq\frac{1}{a}
		\quad\text{if and only if}\quad
		\lambda_iv\leq kr.
	\]
	
	A singular GDD satisfies $r=\lambda_1$, and so $f$ is a universal hash function if and only if $v=k$, which means that every block covers the whole point set. Equivalently, $a=1$.
	
	A semi-regular GDD is characterized by the equality $\lambda_2v=kr$. Further, \eqref{eq:GDD_param_rel} and semi-regularity imply $(\lambda_1-\lambda_2)u=\lambda_1-r\leq 0$, whence $\lambda_1\leq\lambda_2$, and so $\lambda_1v\leq kr$.
\end{proof}

We do not have a simple criterion for when regular GDDs induce a universal hash function. Since a regular GDD $D$ satisfies $kr>\lambda_2v$ by definition, one only needs to check whether $kr\geq\lambda_1v$. This is obviously true if $\lambda_1\leq\lambda_2$. If $\lambda_1>\lambda_2$, then some parameter choices result in mosaics whose functional form is a universal hash function, while this is not true for other parameter choices. 

For instance, the regular GDD R1 from Clatworthy's list \cite{Clatworthy} has parameters $v=4,r=4,k=2,\lambda_1=2,\lambda_2=1$, and thus satisfies $kr=8=\lambda_1v$. Since it is resolvable, an application of Theorem \ref{thm:gen_constr} gives a mosaic of regular GDDs whose functional form is a universal hash function.

On the other hand, the regular GDD R2 from \cite{Clatworthy} has parameters $v=4,r=5,k=2,\lambda_1=3,\lambda_2=1$, hence $kr=10<12=\lambda_1v$. This GDD is resolvable as well, and the functional form of the resulting mosaic is not a universal hash function.

We conclude from Lemma \ref{lem:rel_to_UHF} that not all of the functions constructed in Section \ref{subsect:mos_exs} are universal hash functions. The mosaics of singular GDDs from the family $\mc M^{(3)}$ have functional forms which are not universal hash functions. Similarly, there exist universal hash functions which cannot be decomposed as a mosaic of BIBDs or GDDs. For instance, the optimally universal hash function induced by the resolvable BIBD $AG_{t-1}(t,q)$ (i.e., where the sum of the induced mosaic is $AG_{t-1}(t,q)$) does not have the additional substructure we require from the security functions in this paper.

\subsubsection{Orthogonal arrays.} A $v\times b$ array $M$ with entries from the alphabet $\mc A$ is called a $(b,v,a)$ \textit{orthogonal array} if every $2\times b$ subarray of $M$ contains each pair of entries $(\alpha,\alpha')$ from $\mc A$ exactly $\lambda=b/a^2$ times as a column. 

If we denote the set of rows by $\mc X$ and the set of columns by $\mc S$, then an orthogonal array gives rise to a function $f:\mc X\times\mc S\to\mc A$ which associates to the pair $(x,s)$ the symbol from $\mc A$ which is at the intersection of column $s$ with row $x$. By definition, $f$ satisfies for distinct $x,x'\in\mc X$ and for any $\alpha,\alpha'\in\mc A$
\[
	\lvert\{s:f(x,s)=\alpha,f(x',s)=\alpha'\}\rvert=\lambda.
\]
This means that $f$ is an $\varepsilon$-\textit{almost strongly universal hash function} for $\varepsilon=\lambda a/b$ \cite{Stinson_UH_AC}. In particular,
\begin{equation}\label{eq:OA-two}
	\lvert\{s:f(x,s)=f(x',s)=\alpha\}\rvert=\lambda.
\end{equation}
Moreover, if we set $r=a\lambda$, then 
\[
	\lvert\{s:f(x,s)=\alpha\}\rvert=r.
\]
It is not in general the case that also
\begin{equation}\label{eq:OA-reg}
	\lvert\{x:f(x,s)=\alpha\}\rvert
\end{equation}
is constant in $s$ and $\alpha$. 

Assume \eqref{eq:OA-reg} is constant in $s$ and $\alpha$ and denote this number by $k$. Then $M$ gives a mosaic of $(v,k,\lambda)$ BIBDs with functional form $f$. 

\begin{lem}
	If $M$ is a mosaic of BIBDs induced by an orthogonal array, then $a=1$.
\end{lem}

\begin{proof}
	From $r=a\lambda$ we conclude $rk=v\lambda$. Then \eqref{eq:BIBD_param_rel} gives $r=\lambda$, hence $v=k$.
\end{proof}

\begin{cor}
	There does not exist any nontrivial orthogonal array for which \eqref{eq:OA-reg} is constant in $\alpha$ and $s$.
\end{cor}

\section{Semantic security from mosaics of combinatorial designs}\label{sect:sec}

\subsection{Distances and divergences}\label{subsect:rel_def}

The degree of semantic security offered by a security function when applied to a wiretap channel or in privacy amplification can be measured using various distances, divergences and entropies of probability measures. 

Let $P,Q$ be probability distributions on a finite set $\mc Z$. The \textit{total variation distance} of $P$ and $Q$ is
\[
    \lVert P-Q\rVert=\sum_z\lvert P(z)-Q(z)\rvert.
\]
This is a metric on the space of probability measures on $\mc Z$. The $\chi^2$ divergence
\[
    \chi^2(P,Q)=\sum_{z:Q(z)>0}Q(z)\left(\frac{P(z)}{Q(z)}-1\right)^2
\]
satisfies
\begin{equation}\label{eq:TV_vs_chi2}
    \lVert P-Q\rVert\leq\sqrt{\chi^2(P,Q)}+P(\{z:Q(z)=0\}),
\end{equation}
which is an immediate consequence of Cauchy-Schwarz. The \textit{Kullback-Leibler divergence} of $P$ and $Q$ is given by
\[
    D(P\Vert Q)=\begin{cases}
                    \sum_zP(z)\log\frac{P(z)}{Q(z)} & \text{if }P(\{z:Q(z)=0\})=0,\\
                    +\infty & \text{else},
                \end{cases}
\]
and the \textit{R\'enyi $2$-divergence} by 
\[
    D_2(P\Vert Q)=
    \begin{cases}
        \log\sum_z\frac{P(z)^2}{Q(z)} & \text{if }P(\{z:Q(z)=0\})=0,\\
        +\infty & \text{else}.
    \end{cases}
\]
They are nonnegative and related by \cite{vEH_divergence}
\begin{equation}\label{eq:Renyi_vs_KL}
    D(P\Vert Q)\leq D_2(P\Vert Q).
\end{equation}

It is a straightforward calculation to show that if $D_2(P\Vert Q)<\infty$, then 
\begin{equation}\label{eq:chi2_vs_D2}
    \chi^2(P,Q)=\exp\bigl(D_2(P\Vert Q)\bigr)-1.
\end{equation}

We also introduce averaged versions of these divergences. If $W:\mc X\to\mc Z$ is a channel, and additionally $P$ is a probability distribution on $\mc X$ and $Q$ on $\mc Z$, then we set 
\[
    D(W\Vert Q\vert P)=\sum_{x\in\mc X}P(x)D(W(\,\cdot\,\vert x)\Vert Q)
\]
and
\[
    D_2(W\Vert Q\vert P)=\log\sum_{x\in\mc X}P(x)\exp\bigl(D_2(W(\,\cdot\,\vert x)\Vert Q)\bigr).
\]

Let $X,Y$ be discrete random variables with joint distribution $P_{XY}$. Denote the marginal distributions by $P_X$ and $P_Y$ and the conditional distribution of $Y$ given the event $X=x$ by $P_{Y\vert X=x}$. Then the \textit{mutual information} of $X$ and $Y$ is defined by
\[
    I(X\wedge Y)=\sum_xP_X(x)D(P_{Y\vert X=x}\Vert P_Y)=D(P_{Y\vert X}\Vert P_Y\vert P_X).
\]

The bounds obtained in the privacy amplification scenario involve \textit{R\'enyi 2-entropy}, which for a random variable $X$ on $\mc X$ is defined as
\[
    H_2(X)=-\log\sum_xP_X(x)^2.
\]

\subsection{Wiretap channel}

Let $f:\mc X\times\mc S\to\mc A$ be the functional form of a mosaic $(D_\alpha)_{\alpha\in\mc A}$ of $(v,k,r)$ tactical configurations and let $W:\mc X\to\mc Z$ be a wiretap channel. Assume that the confidential messages to be transmitted are represented by the random variable $A$ on $\mc A$. The random seed is represented by $S$, uniformly distributed on $\mc S$ and independent of $A$. Application of the randomized inverse of $f$ determines the random input $X$ to $W$, and the random output of $W$ seen by Eve is denoted by $Z$. The joint probability distribution of these four random variables is
\begin{equation}\label{eq:WT_joint_distr}
    P_{ZXSA}(z,x,s,\alpha)=\frac{1}{bk}w(z\vert x)N_\alpha(x,s)P_A(\alpha),
\end{equation}
where $N_\alpha$ is the incidence matrix of $D_\alpha$.

The two security metrics by which we measure the degree of security offered by $f$ for  $W$ are defined in terms of the joint distribution of $Z,S$ and $A$ with a worst-case choice of $A$. The first security metric is defined as the mutual information between the message $A$ and the eavesdropper's information $Z,S$, maximized over all possible message distributions,
\begin{equation}\label{eq:def_muti_semsec}
    \max_{P_A}I(A\wedge Z,S).
\end{equation}
The best case would be that Eve's observations are independent of the message, no matter what the message distribution is, in which case the mutual information would vanish. This is not achievable in general, even for a fixed message distribution. Instead, we try to make the maximum in \eqref{eq:def_muti_semsec} as small as possible. Like the other security criteria defined below, the requirement that \eqref{eq:def_muti_semsec} be small does not make any assumptions on Eve's computing power. Thus we aim for \textit{unconditional security}.

\begin{rem}
	For the \textit{strong secrecy} criterion mentioned in Section \ref{subsect:comparison}, it is assumed that the distribution $P_A$ is fixed, so that only the corresponding $I(A\wedge Z,S)$ has to be small. Usually, one takes $A$ to be uniformly distributed on $\mc A$.
\end{rem}

In order to formulate the upper bound for \eqref{eq:def_muti_semsec}, we need to introduce additional notation. If $\mc U$ is a finite set and $R:\mc U\to\mc X$ a channel, then the usual matrix product $RW$ of the stochastic matrices $R$ and $W$ gives the channel with input alphabet $\mc U$ and output alphabet $\mc Z$ resulting from concatenating $R$ and $W$. If $P$ is a probability measure on $\mc X$, then this also defines the probability measure $PW$ on $\mc Z$ by regarding $P$ as a channel with a single row. 

The uniform distribution on any set $\mc X$ is denoted by $P_{\mc X}$. Also, recall R\'enyi 2-divergence defined in Subsection \ref{subsect:rel_def}. 

\begin{thm}\label{thm:WT_muti_sec}
\begin{enumerate}
    \item Let $W:\mc X\to\mc Z$ be a wiretap channel and let $f:\mc X\times\mc S\to\mc A$ be the functional form of a mosaic of $(v,k,\lambda)$ BIBDs. Then
    \begin{align*}
        &\max_{P_A}\exp\bigl(I(A\wedge Z,S)\bigr)
        \leq \left(1-\frac{r-\lambda}{kr}\right)
        +\frac{r-\lambda}{kr}\exp\bigl(D_2(W\Vert P_{\mc X}W\vert P_{\mc X})\bigr).
    \end{align*}
    
    \item Let $W:\mc X\to\mc Z$ be a wiretap channel and let $f:\mc X\times\mc S\to\mc A$ be the functional form of a mosaic of $(u,m,k,\lambda_1,\lambda_2)$ GDDs with a common point class partition $\Pi=\{\mc X_1,\ldots,\mc X_m\}$. Let $P_\Pi$ be the uniform distribution on $\Pi$ and $R_\Pi:\Pi\to\mc X$ the channel which associates to an element $\mc X_j$ of $\Pi$ the uniform distribution on $\mc X_j$. Then
    \begin{align*}
        &\max_{P_A}\exp\bigl(I(A\wedge Z,S)\bigr)\\
        &\leq \left(1-\frac{(r-\lambda_1)+(\lambda_1-\lambda_2)u}{kr}\right)
        +\frac{(\lambda_1-\lambda_2)u}{kr}\exp\bigl(D_2(R_{\Pi}W\Vert P_{\mc X}W\vert P_{\Pi})\bigr)\\
        &\qquad+\frac{r-\lambda_1}{kr}\exp\bigl(D_2(W\Vert P_{\mc X}W\vert P_{\mc X})\bigr).
    \end{align*}
\end{enumerate}
\end{thm}

This theorem is proved in Section \ref{sect:sec_proofs}. The main observation is Proposition \ref{prop:WT_muti_per_des}, which both for the BIBD and the GDD case states equality between $\exp(D_2(P_{Z\vert S,A=\alpha}\Vert P_{Z\vert S}\vert P_{\mc S}))$ and the respective upper bounds in the statement. Since this equality for every $\alpha$ only depends on $D_\alpha$, it really is a statement about BIBDs and GDDs.

Clearly, a GDD with $\lambda_1=\lambda_2$ is a BIBD, so the first part of the theorem is implied by the second one. The same holds for Theorems \ref{thm:WT_TV_sec}, \ref{thm:PA_TV_sec} and \ref{thm:PA_muti_sec} below.

An alternative measure of semantic security is formulated in terms of total variation distance. Denote the product of probability distributions $P$ and $Q$ by $PQ$. Then, with the random variables $Z,S,A$ as defined in \eqref{eq:WT_joint_distr}, we would like
\begin{equation}\label{eq:def_TV_semsec}
    \max_{P_A}\lVert P_{ZSA}-P_{ZS}P_A\rVert
\end{equation}
to be small. If it equals zero, then the eavesdropper's observations are independent of the message, for all possible message distributions.

\begin{thm}\label{thm:WT_TV_sec}
\begin{enumerate}
    \item Let $W:\mc X\to\mc Z$ be a wiretap channel and let $f:\mc X\times\mc S\to\mc A$ be the functional form of a mosaic of $(v,k,\lambda)$ BIBDs. Then
    \begin{align*}
        &\max_{P_A}\lVert P_{ZSA}-P_{ZS}P_A\rVert
        \leq2\left(\frac{(r-\lambda)}{kr}\right)^{1/2}\left(\exp\bigl(D_2(W\Vert P_{\mc X}W\vert P_{\mc X})\bigr)-1\right)^{1/2}.
    \end{align*}
    
    \item  Let $W:\mc X\to\mc Z$ be a wiretap channel and let $f:\mc X\times\mc S\to\mc A$ be the functional form of a mosaic of $(u,m,k,\lambda_1,\lambda_2)$ GDDs with a common point class partition $\Pi$. Define $P_\Pi$ and $R_\Pi$ as in Theorem \ref{thm:WT_muti_sec}. Then
    \begin{align*}
        &\max_{P_A}\lVert P_{ZSA}-P_{ZS}P_A\rVert\\
        &\leq2\Biggl(\frac{r-\lambda_1}{kr}\exp\bigl(D_2(W\Vert P_{\mc X}W\vert P_{\mc X})\bigr)
        +\frac{(\lambda_1-\lambda_2)u}{kr}\exp\bigl(D_2(R_{\Pi}W\Vert P_{\mc X}W\vert P_{\Pi})\bigr)\\
        &\qquad-\frac{(r-\lambda_1)+(\lambda_1-\lambda_2)u}{kr}\Biggr)^{1/2}.
    \end{align*}
\end{enumerate}
\end{thm}

This theorem is also proved in Section \ref{sect:sec_proofs}. It essentially follows from Theorem \ref{thm:WT_muti_sec} and the relations \eqref{eq:TV_vs_chi2} and \eqref{eq:chi2_vs_D2}.

\paragraph{Interpretation.}

The importance of the bounds of Theorems \ref{thm:WT_muti_sec} and \ref{thm:WT_TV_sec} is that they show how much randomness $k$ is sufficient in the randomized inverse in order to obtain a desired level of semantic security. Since $v$ non-confidential messages can be reliably transmitted to Bob, this transforms into a lower bound on the number $a$ of confidential messages. 

The bounds of Theorems \ref{thm:WT_muti_sec} and \ref{thm:WT_TV_sec} can be improved by ``smoothing'' $W$. This means that the outputs of $W$ are restricted to being ``typical'', i.e., outputs of low probability are cut off. This idea goes back to Renner and Wolf \cite{RW_smoothing}. By smoothing, the conditional divergences can be reduced substantially at the cost of a small additive term in each bound. After smoothing, the channel will in general not be stochastic any more, but only substochastic. The proofs of the theorems remain valid for substochastic channels since they only use the nonnegativity of the entries of $W$. All that needs to be done is to generalize the R\'enyi divergences to substochastic channels like in \cite{BRI}.

The bounds can be evaluated by comparing them with the benchmark cases of memoryless discrete and Gaussian wiretap channels (see \cite{BlochBarros} or \cite{BRI} for a definition). These wiretap channels actually are families $\{W_n:n\geq 1\}$ of channels; the parameter $n$ indicates the \textit{blocklength}. For these channels, a sequence of security codes achieves \textit{asymptotic optimality} as the blocklength goes to infinity if the largest possible asymptotic communication rate for confidential message transmission, the \textit{secrecy capacity}, is achieved subject to the condition that either \eqref{eq:def_muti_semsec} or \eqref{eq:def_TV_semsec} goes to zero.

Theorems \ref{thm:WT_muti_sec} and \ref{thm:WT_TV_sec} show that security functions given by suitable mosaics of BIBDs or of semi-regular GDDs achieve asymptotic optimality when applied to memoryless discrete or Gaussian wiretap channels after smoothing each $W_n$. This holds even if the channel between Alice and Bob is not perfect, in which case the $W_n$ are concatenations of an encoder and a memoryless channel. For the proof, one proceeds like in \cite{BRI}. Functional forms of block rate optimal mosaics of singular GDDs turn out to be suboptimal security functions, as discussed below.

We would like to stress, however, that the theorems hold without any further structural assumptions on the channel $W$. For a targeted level of security and a given channel, they can be used to determine an achievable communication rate at which confidential messages can be sent through the channel using an efficiently computable security code.

Note that both in Theorem \ref{thm:WT_muti_sec} and Theorem \ref{thm:WT_TV_sec}, the wiretap channel enters into the upper bounds only through the conditional R\'enyi 2-divergences. This gives some robustness against channel variations or limited channel knowledge.

\paragraph{The bounds in the GDD case.}

Assume that $N$ is the incidence matrix of a $(u,m,k,\lambda_1,\lambda_2)$ GDD and $w\in\mbb R^{\mc X}$ a nonnegative vector. Set $\lambda_\maxsub=\max\{\lambda_1,\lambda_2\}$. Then
\begin{equation}\label{eq:GDD_inc_ub}
	w^TNN^Tw\leq(r-\lambda_\maxsub)w^Tw+\lambda_\maxsub(w^Tj)^2.
\end{equation}
In the proofs of the GDD cases of Theorems \ref{thm:WT_muti_sec} and \ref{thm:WT_TV_sec}, the relation \eqref{eq:GDD_inc_matr} is used with equality. By using \eqref{eq:GDD_inc_ub} instead of \eqref{eq:GDD_inc_matr}, one obtains an upper bound of the same form as that obtained in the BIBD case of the theorems, with $\lambda$ replaced by $\lambda_\maxsub$. Since the point class decomposition of $\mc X$ associated with the applied mosaic of GDDs will not in general have any special relation to the channel, using this looser upper bound might save the work of estimating the additional R\'enyi divergence or entropy and give a bound which, for the benchmark cases and for mosaics of BIBDs or of semi-regular GDDs, is asymptotically equivalent to the one appearing in the theorems. 

The GDD bounds of Theorems \ref{thm:WT_muti_sec} and \ref{thm:WT_TV_sec} can also be simplified without using the upper bound \eqref{eq:GDD_inc_ub} by taking the type of the members of the mosaic $M=(D_\alpha)_{\alpha\in\mc A}$ into consideration. 

In the case where the members of $M$ are singular GDDs, every $D_\alpha$ is induced by a BIBD $D_\alpha^*$. Since the point class partitions of all $D_\alpha$ are the same, all $D_\alpha^*$ have the same parameters $v^*,k^*,\lambda^*$ and form a mosaic of BIBDs. The coefficients of $D_2(W\Vert P_{\mc X}W\vert P_{\mc X})$ vanish, hence only the divergence involving the point class partition is relevant. In Theorem \ref{thm:WT_muti_sec}, the two nonzero coefficients have the form
\begin{equation}\label{eq:WT_singular_coeffs}
1-\frac{r^*-\lambda^*}{k^*r^*}\quad\text{and}\quad\frac{r^*-\lambda^*}{k^*r^*}.
\end{equation}
In Theorem \ref{thm:WT_TV_sec}, both remaining coefficients equal $(r^*-\lambda^*)/k^*r^*$. 

Semi-regular GDDs satisfy $rk=\lambda_2v$. Hence if $M$ consists of semi-regular GDDs, then the three coefficients in Theorem \ref{thm:WT_muti_sec}, in the order of their appearance, equal
\begin{equation}\label{eq:WT_semireg_coeffs}
1,\quad -\frac{r-\lambda_1}{kr},\quad\frac{r-\lambda_1}{kr}.
\end{equation}
For the case where $\lambda_1=0$, in particular, in the case of transversal designs, the same coefficients become
\[
1,\quad -\frac{1}{k},\quad\frac{1}{k}.
\]
The coefficients obtain a similarly simple form in Theorem \ref{thm:WT_TV_sec}.

\paragraph{Suboptimality of singular GDDs.}

When applied in Theorems \ref{thm:WT_muti_sec} and \ref{thm:WT_TV_sec}, approximately block rate optimal mosaics of singular GDDs with a small color rate and a sufficiently large point set achieve strictly lower color rates than mosaics of BIBDs or of semi-regular GDDs at the same security level. In particular, they turn out to be asymptotically suboptimal in the case of memoryless discrete or Gaussian wiretap channels, where the size of the point set goes to infinity with increasing blocklength. This means that asymptotically optimal sequences of security functions given by mosaics of BIBDs or GDDs for these channels have block rates at least 1.

We only discuss Theorem \ref{thm:WT_muti_sec} here, the situation is analogous in Theorem \ref{thm:WT_TV_sec}. We begin with the following simple lemma which is the basis of our discussion.

\begin{lem}\label{lem:divs_comp}
	For a wiretap channel $W:\mc X\to\mc Z$ and a partition $\Pi=\{\mc X_1,\ldots,\mc X_m\}$ of $\mc X$ into sets of size $u$, it holds that
	\[
	D_2(W\Vert P_{\mc X}W\vert P_{\mc X})-\log u
	\leq D_2(R_\Pi W\Vert P_{\mc X}W\vert P_\Pi)
	\leq D_2(W\Vert P_{\mc X}W\vert P_{\mc X}).
	\]
	Equality is possible on both sides. It holds on the left-hand side if and only if for every $z\in\mc Z$ and $1\leq i\leq m$, there exists at most one $x\in\mc X_i$ such that $w(z\vert x)>0$. Equality holds on the right-hand side if and only if for every $z\in\mc Z$ and every $1\leq i\leq m$, the entries $w(z\vert x)$ are constant for $x$ ranging over $\mc X_i$.
\end{lem}

If one applies Theorem \ref{thm:WT_muti_sec} with a mosaic of semi-regular GDDs, then one sees from \eqref{eq:WT_semireg_coeffs}  that a security level $\max_{P_A}I(A\wedge Z,S)$ smaller than $\delta>0$ is achieved by choosing $\log k$ equal to $D_2(W\Vert P_{\mc X}W\vert P_{\mc X})+\log(1/\delta)$. This results in the color rate
\[
	\tilde\varrho
	=1-\frac{D_2(W\Vert P_{\mc X}W\vert P_{\mc X})+\log(1/\delta)}{\log v}.
\]
The same holds in the simpler situation of mosaics of BIBDs.  

Now assume that $\tilde\varrho<1/2$. By Section \ref{subsect:bro}, the only possibility to achieve a security level smaller than $\delta$ for the same channel $W$ with an approximately block rate optimal mosaic could be a mosaic $M$ of singular GDDs which is the $u$-fold multiple of a mosaic $M^*$ of block rate optimal BIBDs and of color rate $\varrho^*$. When Theorem \ref{thm:WT_muti_sec} is applied with the security function determined by $M$, the $D_2(W\Vert P_{\mc X}W\vert P_{\mc X})$ term vanishes in the upper bound of Theorem \ref{thm:WT_muti_sec}. By \eqref{eq:WT_singular_coeffs}, a security level smaller than $\delta$ is achieved by choosing $\log k^*$ equal to $D_2(R_\Pi W\Vert P_{\mc X}W\vert P_\Pi)+\log(1/\delta)$, and without any further information about the channel, this latter expression can be as large as $D_2( W\Vert P_{\mc X}W\vert P_{\mc X})+\log(1/\delta)$ by Lemma \ref{lem:divs_comp}.

 For the color rate $\varrho$ of $M$, this means that
\begin{equation}\label{eq:varrho-bounds}
	\varrho
	=1-\frac{\log k}{\log v}
	=1-\frac{\log k^*+\log u}{\log v}
	\leq\tilde\varrho-\frac{\log u}{\log v}.
\end{equation}
This is at most $\tilde\varrho$. In fact, for fixed $\tilde\varrho$, it is easy to see that $\log u/\log v$ is bounded from below for large $v$. This is because the approximate block rate optimality of $M$ requires $\varrho^*$ to be at least $\varrho_0(v^*,k^*)$, which tends to $1/2$ as $v^*$ grows. And if $v^*$ is kept small, then $u$ necessarily has to be large.

The loss of color rate as in \eqref{eq:varrho-bounds} can be avoided if one knows that equality is satisfied in the left-hand inequality of Lemma \ref{lem:divs_comp} for a certain partition $\Pi$. However, an application of this in the security bounds would require knowledge of $D_2(R_\Pi W\Vert P_{\mc X}W\vert P_\Pi)$ and the adaptation of the point class partition of the GDDs to that of the wiretap channel, which is not necessary in the case of mosaics of BIBDs or of semi-regular GDDs.

\subsection{Privacy amplification}

Now we turn to privacy amplification. Assume that the random variable $X$ is shared by Alice and Bob and that Eve observes a random variable $Z$ correlated with $X$. Without loss of generality, we assume that $P_Z(z)>0$ for all $z\in\mc Z$. Moreover, Alice and Bob both are given the functional form $f:\mc X\times\mc S\to\mc A$ of a mosaic $(D_\alpha)_{\alpha\in\mc A}$ of $(v,k,r)$ tactical configurations. In order to generate a secret key, Alice and Bob observe a realization $x$ of $X$, choose a seed $s\in\mc S$ uniformly at random, and take $\alpha=f(x,s)$ as the secret key. Denote the random variable generated by applying $f$ as described above by $A$. The joint distribution of $X,Z,S$ and $A$ is
\begin{equation}\label{eq:PA_joint_distr}
    P_{XZSA}(x,z,s,\alpha)=\frac{1}{b}P_{XZ}(x,z)N_\alpha(x,s),
\end{equation}
where $N_\alpha$ is the incidence matrix of $D_\alpha$. The key $A$ should be nearly uniformly distributed on $\mc A$ and semantically secure with respect to Eve's observation. The first condition is satisfied perfectly.

\begin{lem}\label{lem:A_unif}
    The distribution of $A$ is uniform on $\mc A$.
\end{lem}

\begin{proof}
    Note that $N_\alpha j=rj$, where $j$ denotes the all-ones vector of appropriate dimension. Hence, considering $P_X$ as a vector in $\mbb R^{\mc X}$ and using \eqref{eq:bireg},
    \[
    	P_A(\alpha)
        =\sum_{x,z,s}P_{XZSA}(x,z,s,\alpha)
        =\frac{1}{b}\sum_{x,s}P_X(x)N_\alpha(x,s)j(s)
        =\frac{1}{b}P_X^TN_\alpha j
        =\frac{r}{b}P_X^Tj
        =\frac{k}{v}=\frac{1}{a}.
    \]
\end{proof}

For semantic security, we can again use total variation distance or mutual information as the security measure. One equivalent formulation of semantic security is the indistinguishability of two possible realizations of the secret. In terms of total variation distance, this means that for any two distinct $\alpha,\alpha'\in\mc A$, one wants
\[
    \lVert P_{ZS\vert A=\alpha}-P_{ZS\vert A=\alpha'}\rVert
\]
to be uniformly small. By the triangle inequality, this is true if
\begin{equation}\label{eq:TV_indist}
    \lVert P_{ZS\vert A=\alpha}-P_ZP_{\mc S}\rVert
\end{equation}
is small, uniformly in $\alpha\in\mc A$. 

For any point class partition $\Pi=\{\mc X_1,\ldots,\mc X_m\}$ of $\mc X$, we define the random variable $X_\Pi$ whose conditional distribution given $Z$ is
\[
    P_{X_\Pi\vert Z}(i\vert z)=P_{X\vert Z}(\mc X_i\vert z).
\]
Then we have the following result.

\begin{thm}\label{thm:PA_TV_sec}
\begin{enumerate}
    \item Let $P_{XZSA}$ be the joint distribution \eqref{eq:PA_joint_distr} generated by the functional form of a mosaic of $(v,k,\lambda)$ BIBDs. Then
    \[
        \max_{\alpha\in\mc A}\lVert P_{ZS\vert A=\alpha}-P_ZP_{\mc S}\rVert\leq\left(\frac{r-\lambda}{r}\right)^{1/2}\left(a2^{-\min_zH_2(X\vert Z=z)}-\frac{1}{k}\right)^{1/2}
    \]
    \item Let $P_{XZSA}$ be the joint distribution \eqref{eq:PA_joint_distr} generated by the functional form of a mosaic of $(u,m,k,\lambda_1,\lambda_2)$ GDDs with a common point class partition $\Pi$. Then
    \begin{align*}
        \max_{\alpha\in\mc A}\lVert P_{ZS\vert A=\alpha}-P_ZP_{\mc S}\rVert
        &\leq\max_{z\in\mc Z}\Biggl\{\frac{a(r-\lambda_1)}{r}2^{-H_2(X\vert Z=z)}+\frac{a(\lambda_1-\lambda_2)}{r}2^{-H_2(X_\Pi\vert Z=z)}\\
        &\qquad\qquad-\frac{(r-\lambda_1)+(\lambda_1-\lambda_2)u}{kr}\Biggr\}^{1/2}.
    \end{align*}
\end{enumerate}
\end{thm}

This is proved in Section \ref{sect:sec_proofs} as a consequence of the next theorem. 

If we prefer to measure the indistinguishability of key values with respect to Kullback-Leibler divergence, we should ensure that there exists a probability measure $Q$ on $\mc Z\times\mc S$ such that $P_{ZS\vert A=\alpha}$ is close to $Q$ in terms of Kullback-Leibler divergence, uniformly in $\alpha\in\mc A$. This is analogous to \eqref{eq:TV_indist}. If we choose $Q=P_ZP_{\mc S}$, then we have the following bound.

\begin{thm}\label{thm:PA_muti_sec}
\begin{enumerate}
    \item Let $P_{XZSA}$ be the joint distribution \eqref{eq:PA_joint_distr} generated by the functional form of a mosaic of $(v,k,\lambda)$ BIBDs. Then
    \begin{align*}
        \max_{\alpha\in\mc A}\exp\bigl(D(P_{ZS\vert A=\alpha}\Vert P_ZP_{\mc S})\bigr)
        &\leq\frac{a(r-\lambda)}{r}2^{-\min_zH_2(X\vert Z=z)}
        +\left(1-\frac{r-\lambda}{kr}\right).
    \end{align*}
    \item Let $P_{XZSA}$ be the joint distribution \eqref{eq:PA_joint_distr} generated by the functional form of a mosaic of $(u,m,k,\lambda_1,\lambda_2)$ GDDs with a common point class partition $\Pi$. Then
    \begin{align*}
        &\max_{\alpha\in\mc A}\exp\bigl(D(P_{ZS\vert A=\alpha}\Vert P_ZP_{\mc S})\bigr)\\
        &\leq\max_{z\in\mc Z}\biggl\{\frac{a(r-\lambda_1)}{r}2^{-H_2(X\vert Z=z)}+\frac{a(\lambda_1-\lambda_2)}{r}2^{-H_2(X_\Pi\vert Z=z)}\\
        &\quad\qquad+\left(1-\frac{(r-\lambda_1)+(\lambda_1-\lambda_2)u}{kr}\right)\biggr\}.
    \end{align*}
\end{enumerate}
\end{thm}

The theorem is proved in Section \ref{sect:sec_proofs}. As in the wiretap case, its core is Proposition \ref{prop:PA_muti_per_des}, proving the equality of $\exp(D_2(P_{S\vert Z=z,A=\alpha}\Vert P_{\mc S}))$ with the $z$-term in the upper bound.

\begin{rem}
	The \textit{strong secrecy} criterion usually applied in information theoretic security for secret key generation  assumes that the adversary's a priori knowledge is restricted to the true key distribution. A security function which establishes semantic security also guarantees strong secrecy, since
	\begin{equation}\label{eq:PA_strongsec}
		I(A\wedge Z,S)\leq\max_{\alpha\in\mc A}D(P_{ZS\vert A=\alpha}\Vert P_ZP_{\mc S}).
	\end{equation}
	
	We prove this inequality. It is straightforward to check that for any pair of random variables $X,Y$ on $\mc X\times\mc Y$ and any probability measure $Q$ on $\mc Y$, one has
	\[
	I(X\wedge Y)=D(P_{XY}\Vert P_XP_Y)=\sum_xP_X(x)D(P_{Y\vert X=x}\Vert Q)-D(P_Y\Vert Q).
	\]
	We use this with $Y=(Z,S),X=A$ and $Q=P_ZP_{\mc S}$. Then 
	\begin{align*}
		I(A\wedge Z,S)
		&=D(P_{ZSA}\Vert P_{ZS}P_{\mc A})
		=\sum_\alpha P_A(\alpha)D(P_{ZS\vert A=\alpha}\Vert P_ZP_{\mc S})-D(P_{ZS}\Vert P_ZP_{\mc S})\\
		&\leq\max_\alpha D(P_{ZS\vert A=\alpha}\Vert P_ZP_{\mc S}).
	\end{align*}
	This shows \eqref{eq:PA_strongsec}.
\end{rem}

\paragraph{Interpretation.}

The interpretation of Theorems \ref{thm:PA_TV_sec} and \ref{thm:PA_muti_sec} is analogous to that of Theorems \ref{thm:WT_muti_sec} and \ref{thm:WT_TV_sec}. The number of interest is $a$, the size of the key space. Theorems \ref{thm:PA_TV_sec} and \ref{thm:PA_muti_sec} give a lower bound on the maximal possible $a$ given a required degree of security, and show that this lower bound is achievable using the functional form of a mosaic of BIBDs or GDDs. 

It is proved in \cite[Corollary 4]{BBCM_PA} that
\[
\exp\bigl(I(A\wedge S\vert Z=z)\bigr)\leq a2^{-\min_zH_2(X\vert Z=z)}+1
\]
if the security function is a universal hash function. The upper bound is very similar to the one proved in the first part of Theorem \ref{thm:PA_muti_sec} for mosaics of BIBDs or of semi-regular GDDs, but only gives strong secrecy. (The conditioning on the event $Z=z$ is also possible in our setting, see \eqref{eq:PA_restr_z}.) It follows that these mosaics yield the same key size as universal hash functions, but resulting in a stronger notion of security and generating a perfectly uniformly distributed key. Mosaics of singular GDDs only involve the $\min_zH_2(X_\Pi\vert Z=z)$ term and are discussed in more detail below.

If Alice and Bob are connected by a public two-way channel without rate constraint, the secret-key capacity in the benchmark case of a memoryless discrete source model can be achieved by a sequential key distillation protocol guaranteeing semantic security, using functional forms of mosaics of BIBDs or suitable GDDs in the privacy amplification step (cf.~\cite[Theorem 4.5]{BlochBarros}).

\paragraph{The bounds in the GDD case.}

By applying \eqref{eq:GDD_inc_ub}, the bounds for the GDD cases of Theorems \ref{thm:PA_TV_sec} and \ref{thm:PA_muti_sec} can be given the same form as the ones for the BIBD case, with $\lambda$ replaced by $\lambda_\maxsub$. 

If the mosaic consists of singular GDDs, then the coefficient of the $H_2(X\vert Z=z)$ term vanishes. The second and third terms in Theorem \ref{thm:PA_muti_sec} are
\[
\frac{a(r^*-\lambda^*)}{r^*}\quad\text{and}\quad-\frac{r^*-\lambda^*}{k^*r^*},
\]
where, like in the wiretap scenario, $k^*,r^*,\lambda^*$ are parameters of the underlying BIBDs. 

In the case of semi-regular GDDs, one has, in the order of their appearance, the three terms
\[
\frac{a(r-\lambda_1)}{r},\quad -\frac{a(r-\lambda_1)}{ur},\quad 0.
\]
In particular, for transversal designs, one obtains
\[
a,\quad-1,\quad 0.
\]

Similar simplifications are possible for the bounds of Theorem \ref{thm:PA_TV_sec}.

\paragraph{Suboptimality of singular GDDs.}

As in the wiretap scenario, mosaics of singular GDDs are suboptimal compared with mosaics of BIBDs or of semi-regular GDDs since they require a larger $k$ in order to achieve a comparable security level. 

The reasons are analogous to those for the wiretap case, based on the inequalities
\begin{equation}\label{eq:PA_entr_comparison}
	H_2(X\vert Z=z)-\log u\leq H_2(X_\Pi\vert Z=z)\leq H_2(X\vert Z=z)
\end{equation}
for any partition $\Pi=\{\mc X_1,\ldots,\mc X_m\}$ of $\mc X$ into sets of size $u$, and any $z\in\mc Z$. The condition for equality in the right-hand inequality is that there exist at most one $x$ per $\mc X_i$ with $P_{X\vert Z}(x\vert z)>0$. On the left-hand side, equality holds if and only if $P_{X\vert Z}(\cdot\vert z)$ is constant on each $\mc X_i$ for every $z$.

With a mosaic of BIBDs or of semi-regular GDDs, a key size $\log a$ approximately equal to $\min_zH_2(X\vert Z=z)+\log(1/\delta)$ gives a security level $\delta$.

Now assume that the security function is given by a mosaic of singular GDDs. If one only knows $\min_zH_2(X\vert Z=z)$, then the largest possible key size $\log a$ by which to guarantee a security level of $\delta$ is $H_2(X\vert Z=z)-\log u+\log(1/\delta)$. The key can be chosen larger if one also knows $\min_zH_2(X_\Pi\vert Z=z)$. However, the same key size as in the case of BIBDs or semi-regular GDDs is achievable only if there exists a partition $\Pi$ such that equality is satisfied on the right-hand side of \eqref{eq:PA_entr_comparison}. If one knows that the joint distribution $P_{XZ}$ has this property for a partition $\Pi$, then a mosaic of singular GDDs incurs no rate loss, but the security function has to be adapted to $\Pi$.

\section{Proofs of the security results}\label{sect:sec_proofs}

\subsection{Proof of Theorems \ref{thm:WT_muti_sec} and \ref{thm:WT_TV_sec}}

We first prove Theorem \ref{thm:WT_muti_sec}. It is sufficient to do the proof for mosaics of GDDs. We start with an upper bound on $\max_{P_A}I(A\wedge Z,S)$ in terms of Kullback-Leibler divergence. The all-ones vector of suitable dimension will be denoted by $j$, and for each $z\in\mc Z$, we let $w_z$ be the $z$-th column of $W$

\begin{lem}\label{lem:WT_muti_prep}
	For every joint distribution \eqref{eq:WT_joint_distr}, 
	\[
		I(A\wedge Z,S)\leq\max_{\alpha\in\mc A} D(P_{Z\vert S,A=\alpha}\Vert P_Z\vert P_{\mc S}),
	\]
	and the right-hand side of this inequality is independent of $P_A$.
\end{lem}

\begin{proof}
	The inequality is the statement of \cite[Corollary 16]{BRI}, whose proof we will just sketch here. The independence of $A$ and $S$ implies $I(A\wedge Z,S)\leq I(A,S\wedge Z)$ using elementary properties of mutual information. The right-hand mutual information can be expressed as
	\[
		\frac{1}{b}\sum_{s\in\mc S}\sum_{\alpha\in\mc A}P_A(\alpha)D(P_{Z\vert S=s,A=\alpha}\Vert P_Z)
		\leq\max_{\alpha\in\mc A}D(P_{Z\vert S,A=\alpha}\Vert P_Z\vert P_{\mc S}).
	\]
	This gives the claimed inequality.
	
	In order to prove that the upper bound is independent of $P_A$, we note that \eqref{eq:WT_joint_distr} and \eqref{eq:bireg} imply
	\begin{equation}\label{eq:WT_Z_unif}
	    P_Z(z)
	    =\frac{1}{bk}\sum_{\alpha\in\mc A}P_A(\alpha)w_z^TN_\alpha j
	    =\frac{1}{v}w_z^Tj
	    =(P_{\mc X}W)(z).
	\end{equation}
	Thus $P_Z$ is independent of $P_A$. Since $P_{\mc S}$ and $P_{Z\vert S,A=\alpha}$ do not depend on $P_A$ either, this proves the lemma.
\end{proof}

Note that, since the eavesdropper also knows $S$, the validity of \eqref{eq:WT_Z_unif} is not enough to guarantee security. 

If we want to use \eqref{eq:BIBD_inc_matr} or \eqref{eq:GDD_inc_matr}, we need to pass from Kullback-Leibler to R\'enyi $2$-divergence. By Lemma \ref{lem:WT_muti_prep} and \eqref{eq:Renyi_vs_KL}, it is sufficient to show that the upper bound of Theorem \ref{thm:WT_muti_sec} is an upper bound for 
\begin{equation}\label{eq:WT_max_alpha_div}
    \max_{\alpha\in\mc A}D_2(P_{Z\vert S,A=\alpha}\Vert P_Z\vert P_{\mc S}).
\end{equation}
$P_{Z\vert S,A=\alpha}$ is fully determined by $N_\alpha$ and $W$. Hence for each of the divergence terms in \eqref{eq:WT_max_alpha_div} it is no longer important that $N_\alpha$ is the incidence matrix of a member of a mosaic. It follows that Theorem \ref{thm:WT_muti_sec} is a consequence of the following equality.

\begin{prop}\label{prop:WT_muti_per_des}
    Let $N$ be the incidence matrix of a $(u,m,k,\lambda_1,\lambda_2)$ GDD with point set $\mc X$, block index set $\mc S$ and point class partition $\Pi$, and let $W:\mc X\to\mc Z$ be a wiretap channel. Define the random variables $Z,X,S$ on $\mc Z\times\mc X\times\mc S$ by
    \begin{equation}\label{eq:WT_single_des_joint_distr}
        P_{ZXS}(z,x,s)=\frac{1}{bk}w(z\vert x)N(x,s).
    \end{equation}
    Then 
    \begin{align*}
        &\exp\bigl(D_2(P_{Z\vert S}\Vert P_Z\vert P_{\mc S})\bigr)\\
        &=\left(1-\frac{(r-\lambda_1)+(\lambda_1-\lambda_2)u}{kr}\right)
        +\frac{(\lambda_1-\lambda_2)u}{kr}\exp\bigl(D_2(R_{\Pi}W\Vert P_{\mc X}W\vert P_{\Pi})\bigr)\\
        &\qquad+\frac{r-\lambda_1}{kr}\exp\bigl(D_2(W\Vert P_{\mc X}W\vert P_{\mc X})\bigr).
    \end{align*}

\end{prop}

\begin{proof}
    As in \eqref{eq:WT_Z_unif}, one shows that $P_Z=P_{\mc X}W$. Since also
    \[
        P_{Z\vert S}(z\vert s)=\frac{(w_z^TN)(s)}{k},
    \]
    we can apply \eqref{eq:GDD_inc_matr} and obtain
    \begin{align*}
        \exp\bigl(D_2(P_{Z\vert S}\Vert P_Z\vert P_{\mc S})\bigr)
        &=\frac{v}{bk^2}\sum_{s\in\mc S}\sum_{z\in\mc Z}\frac{(w_z^TN)(s)^2}{w_z^Tj}\\
        &=\frac{1}{kr}\sum_{z\in\mc Z}\frac{w_z^TNN^Tw_z}{w_z^Tj}\\
        &=\frac{r-\lambda_1}{kr}\sum_{z\in\mc Z}\frac{w_z^Tw_z}{w_z^Tj}+\frac{\lambda_1-\lambda_2}{kr}\sum_{z\in\mc Z}\frac{w_z^TCw_z}{w_z^Tj}+\frac{\lambda_2}{kr}\sum_{z\in\mc Z}w_z^Tj.
    \end{align*}
    Now, observe that 
    \[
        \sum_{z\in\mc Z}\frac{w_z^Tw_z}{w_z^Tj}
        =\frac{1}{v}\sum_{x\in\mc X}\sum_{z\in\mc Z}\frac{w(z\vert x)^2}{(P_{\mc X}W)(z)}
        =\exp\bigl(D_2(W\Vert P_{\mc X}W\vert P_{\mc X})\bigr).
    \]
    In the second summand, we have
    \[
        \sum_{z\in\mc Z}\frac{w_z^TCw_z}{w_z^Tj}
        =\frac{u}{m}\sum_{i=1}^m\sum_{z\in\mc Z}\frac{\left(u^{-1}\sum_{x\in\mc X_i}w(z\vert x)\right)^2}{(P_{\mc X}W)(z)}
        =u\exp\bigl(D_2(R_\Pi W\Vert P_{\mc X}W\vert P_\Pi)\bigr).
    \]
    For the third summand, we observe that $\sum_zw_z^Tj=v$ and 
    \begin{equation}\label{eq:GDD_param_rel_angew}
        \lambda_2v=kr-(r-\lambda_1)-(\lambda_1-\lambda_2)u,
    \end{equation}
    which follows from \eqref{eq:GDD_param_rel}. Inserting all this above yields the claimed equality.
\end{proof}

Turning to the proof of Theorem \ref{thm:WT_TV_sec}, we first state the following simple analog of Lemma \ref{lem:WT_muti_prep}.

\begin{lem}[\cite{Cs_almost_indep}, Lemma 2]
    \begin{equation}\label{eq:WT_TV_prep}
        \lVert P_{ZSA}-P_{ZS}P_A\rVert\leq 2\lVert P_{ZSA}-P_ZP_{\mc S}P_A\rVert.
    \end{equation}
\end{lem}

By \eqref{eq:WT_Z_unif}, $P_Z(z)=0$ only if $z$ is not reachable with positive probability from any input of $W$. Thus 
\[
	P_{ZS\vert A=\alpha}(\{(z,s):P_Z(z)P_{\mc S}(s)=0\})=0.
\]
Hence one can apply \eqref{eq:TV_vs_chi2} and \eqref{eq:chi2_vs_D2} to upper-bound the right-hand side of \eqref{eq:WT_TV_prep} by
\begin{align*}
    2\max_{\alpha\in\mc A}\lVert P_{ZS\vert A=\alpha}-P_ZP_{\mc S}\rVert
    &\leq2\max_{\alpha\in\mc A}\sqrt{\chi^2( P_{ZS\vert A=\alpha}-P_ZP_{\mc S})}\\
    &=2\max_{\alpha\in\mc A}\sqrt{\exp\bigl(D_2(P_{ZS\vert A=\alpha}\Vert P_ZP_{\mc S})\bigr)-1}\\
    &=2\max_{\alpha\in\mc A}\sqrt{\exp\bigl(D_2(P_{Z\vert S,A=\alpha}\Vert P_Z\vert P_{\mc S})\bigr)-1}.
\end{align*}
Theorem \ref{thm:WT_TV_sec} now follows from Proposition \ref{prop:WT_muti_per_des}.

\subsection{Proof of Theorems \ref{thm:PA_TV_sec} and \ref{thm:PA_muti_sec}}

We start by proving Theorem \ref{thm:PA_muti_sec}. Define the $\mbb R^{\mc X}$-vector $p_z$ by
\[
    p_z(x)=P_{XZ}(x,z).
\]
From \eqref{eq:PA_joint_distr}, Lemma \ref{lem:A_unif} and \eqref{eq:bireg}, it follows that
\begin{equation}\label{eq:Eve_given_key}
    P_{Z\vert A}(z\vert\alpha)
    =a\sum_{x,s}P_{XZSA}(x,z,s,\alpha)
    =\frac{a}{b}p_z^TN_\alpha j
    =\frac{ar}{b}p_z^Tj
    =p_z^Tj.
\end{equation}
In particular, $Z$ is independent of $A$. (Of course, since the eavesdropper also knows $S$, this is not yet enough to guarantee security.) A straightforward computation gives
\begin{equation}\label{eq:PA_restr_z}
    D(P_{ZS\vert A=\alpha}\Vert P_ZP_{\mc S})
    \leq\max_{z\in\mc Z}D(P_{S\vert Z=z,A=\alpha}\Vert P_{\mc S}).
\end{equation}
As in the wiretap case, one passes to R\'enyi 2-divergence, and so it remains to bound
\[
    D_2(P_{S\vert Z=z,A=\alpha}\Vert P_{\mc S}),
\]
uniformly in $z$ and $\alpha$. 

We compute $P_{S\vert Z=z,A=\alpha}$ as follows. Recall the assumption that $P_Z(z)>0$ for all $z\in\mc Z$. Let $\alpha\in\mc A$ and $s\in\mc S$. The uniform distribution of $A$ and \eqref{eq:Eve_given_key} imply that $P_{ZA}(z,\alpha)=a^{-1}p_z^Tj$. Hence, again applying \eqref{eq:bireg},
\begin{align}\label{eq:PA_cond_S_distr}
    P_{S\vert Z=z,A=\alpha}(s)
    &=\frac{a\sum_xP_{XZSA}(x,z,s,\alpha)}{p_z^Tj}\notag\\
    &=\frac{a}{bp_z^Tj}\sum_xP_{XZ}(x,z)N_\alpha(x,s)
    =\frac{(p_z^TN_\alpha)(s)}{rp_z^Tj}.
\end{align}
This only depends on the incidence matrix $N_\alpha$, and so as in the wiretap case, we can reduce the proof of Theorem \ref{thm:PA_muti_sec} to a proposition which holds for GDDs without any reference to mosaics.

\begin{prop}\label{prop:PA_muti_per_des}
    Let $N$ be the incidence matrix of a $(u,m,k,\lambda_1,\lambda_2)$ GDD with point set $\mc X$, block index set $\mc S$ and point class partition $\Pi$. Let $\mc Z$ be a finite set and define the random variables $X,Z,S$ on $\mc X,\mc Z,\mc S$, respectively, by their joint distribution
    \begin{equation}\label{eq:PA_single_des_cond_distr}
        P_{XZS}(x,z,s)=\frac{1}{r}P_{XZ}(x,z)N(x,s).
    \end{equation}
    Then
    \begin{align*}
        \exp\bigl(D_2(P_{S\vert Z=z}\Vert P_{\mc S})\bigr)
        &=\frac{v(r-\lambda_1)}{kr}2^{-H_2(X\vert Z=z)}+\frac{v(\lambda_1-\lambda_2)}{kr}2^{-H_2(X_\Pi\vert Z=z)}\\
        &\qquad+\left(1-\frac{(r-\lambda_1)+(\lambda_1-\lambda_2)u}{kr}\right).
    \end{align*}
\end{prop}

\begin{proof}
As in \eqref{eq:PA_cond_S_distr}, it holds that $P_{S\vert Z}(s\vert z)=(p_z^TN)(s)/rp_z^Tj$. Using \eqref{eq:GDD_inc_matr}, one obtains
\begin{align*}
    &\exp\bigl(D_2(P_{S\vert Z=z}\Vert P_{\mc S})\bigr)\\
    &=\sum_s\frac{b(p_z^TN)(s)^2}{r^2(p_z^Tj)^2}\\
    &=\frac{vp_z^TN N^Tp_z}{kr(p_z^Tj)^2}\\
    &=\frac{v}{kr(p_z^Tj)^2}\left((r-\lambda_1)p_z^Tp_z+(\lambda_1-\lambda_2)p_z^TCp_z+\lambda_2(p_z^Tj)^2\right)\\
    &=\frac{v(r-\lambda_1)}{kr}2^{-H_2(X\vert Z=z)}+\frac{v(\lambda_1-\lambda_2)}{kr}2^{-H_2(X_\Pi\vert Z=z)}+\frac{v\lambda_2}{kr}.
\end{align*}
    The proof is complete upon replacing the last summand using \eqref{eq:GDD_param_rel_angew}.
\end{proof}

This completes the proof of Theorem \ref{thm:PA_muti_sec}. 

In order to prove Theorem \ref{thm:PA_TV_sec}, we can appeal to the case where security is measured using divergence, just like in the wiretap case. It is a straightforward computation to show that
\begin{align*}
    \lVert P_{ZS\vert A=\alpha}-P_ZP_{\mc S}\rVert
    &\leq\max_z\lVert P_{S\vert Z=z,A=\alpha}-P_{\mc S}\rVert
\end{align*}
for all $\alpha\in\mc A$. Using \eqref{eq:TV_vs_chi2} and \eqref{eq:chi2_vs_D2}, we see that Theorem \ref{thm:PA_TV_sec} follows from Proposition \ref{prop:PA_muti_per_des}.

\begin{rem}\label{rem:generalization}
    Let $N$ be the incidence matrix of a $(v,k,r)$ tactical decomposition for which there exist nonnegative numbers $c$ and $d$ such that
    \begin{equation}\label{eq:gen_inc_matr_ub}
        w^TNN^Tw\leq cw^Tw+d(w^Tj)^2
    \end{equation}
    for all nonnegative vectors $w$. Propositions \ref{prop:WT_muti_per_des} and \ref{prop:PA_muti_per_des} can be generalized for such matrices, with an inequality instead of an equality.
    
    Let $W:\mc X\to\mc Z$ be a wiretap channel and define the random variables $Z,X,S$ on $\mc Z\times\mc X\times\mc S$ as in \eqref{eq:WT_single_des_joint_distr}. Proceeding as in the proof of Proposition \ref{prop:WT_muti_per_des}, one can show that
    \[
        \exp\bigl(D_2(P_{Z\vert S}\Vert P_Z\vert P_{\mc S})\bigr)
        \leq\frac{dv}{kr}+\frac{c}{kr}\exp\bigl(D_2(W\Vert P_{\mc X}W\vert P_{\mc X})\bigr).
    \]
    Similarly, in privacy amplification with source distribution $P_{XZ}$ and with the seed jointly distributed with $Z$ according to \eqref{eq:PA_single_des_cond_distr}, one obtains
    \[
        \exp\bigl(D_2(P_{S\vert Z=z}\Vert P_{\mc S})\bigr)
        \leq\frac{ac}{r}2^{-H_2(X\vert Z=z)}+\frac{ad}{r},
    \]
    proceeding as in the proof of Proposition \ref{prop:PA_muti_per_des}.  
    
    For example, if $NN^T$ has largest eigenvalue $\mu_1$ and second-largest eigenvalue $\mu_2$, then 
    \[
        w^TNN^Tw\leq\mu_2w^Tw+\frac{\mu_1-\mu_2}{v}(w^Tj)^2.
    \]
    Mosaics of such matrices were studied in the wiretap scenario in \cite{BRI}.
    
    Another example of a matrix satisfying \eqref{eq:gen_inc_matr_ub} arises from the incidence matrix of a $(u,m,k,\lambda_1,\lambda_2)$ GDD, see \eqref{eq:GDD_inc_ub}.
    
    Theorems \ref{thm:WT_muti_sec}, \ref{thm:WT_TV_sec}, \ref{thm:PA_TV_sec} and \ref{thm:PA_muti_sec} can also be generalized to mosaics of tactical configurations whose incidence matrices satisfy \eqref{eq:gen_inc_matr_ub}, since the reduction of the theorems to Propositions \ref{prop:WT_muti_per_des} and \ref{prop:PA_muti_per_des} only used that the security functions are functional forms of mosaics of tactical configurations. 
\end{rem}

\section{Explicitness of Denniston's BIBD}\label{sect:maxarc_proof}

Let $t\geq 2$ and $1\leq\ell\leq t$. Set $q=2^t$. Recall that Denniston's design $D$, defined in Section \ref{subsect:mos_exs}, has the point set
\[
    \mc X=\{(x,y)\in\mbb F_q^2:Q(x,y)\in H\},
\]
where $Q$ is an irreducible quadratic form and $H$ a subgroup of $\mbb F_q$ of order $\ell$. We will consider $\mbb F_q$ as a $t$-dimensional vector space over $\mbb F_2$, which makes $H$ an $\ell$-dimensional subspace of $\mbb F_q$. The blocks of $D$ are given by the nontrivial intersections of lines of $AG(2,q)$ with $\mc X$. 

\begin{prop}\label{prop:Denniston_explicit}
	There exists an $H$ such that $D$ is explicit.
\end{prop}

The proof of this proposition will be done in the subsections following below. We first observe that the proposition implies that the mosaic $M^{(2)}_{t,\ell,H}$ whose members are isomorphic to $D$ is explicit. This follows from Theorem \ref{thm:comp_compl} together with the efficiency of addition and subtraction on the cyclic group $\mbb Z_a$, which serves as the color set for the mosaic.

\subsection{Characterization of $\mc X$ and $\mc S$}

Denote by $L_{c,d}=\{(x,cx+d):x\in\mbb F_q\}$ the line in $AG(2,q)$ with \textit{slope} $c\in\mbb F_q$ and \textit{intercept} $d\in\mbb F_q$. This are all lines of $AG(2,q)$ except the ``vertical'' ones with infinite slope, given by $L_{\infty,d}=\{(d,y):y\in\mbb F_q\}$, for any $d\in\mbb F_q$. For these lines, we call $d$ the \textit{intercept}.

For the characterization of $\mc X$ and $\mc S$, we choose $H$ arbitrary. Note that $0\in\mc X$. Thus every line $L_{c,0}$ ($c\in\mbb F_q\cup\{\infty\}$) has nontrivial intersection $\mc X_c$ with $\mc X$. Since any two of these lines only meet in 0, the union of all these $\mc X_c$ has precisely 
\[
    v=1+(2^t+1)(2^\ell-1)=2^{t+\ell}+2^\ell-2^t
\]
elements, and so $\mc X$ must equal the union of all $\mc X_c$ by \eqref{eq:Denniston_card_X}. Now assume $c\in\mbb F_q$. An element $(x,cx)$ of $L_{c,0}$ is contained in $\mc X_c$ if and only if
\[
    x^2(\eta_1+\eta_2c+\eta_3c^2)\in H,
\]
or equivalently, $x^2\in(\eta_1+\eta_2c+\eta_3c^2)^{-1}H$ (the irreducibility of $Q$ ensures that $\eta_1+\eta_2c+\eta_3c^2$ is nonzero). In an analogous way one sees that $(0,y)\in\mc X_\infty$ if and only if $y^2\in\eta_3^{-1}H$. 

\begin{lem}\label{lem:Denniston_X}
    The set $\mc X$ is given by the disjoint union
    \[
        \{(0,0)\}\cup\bigcup_{c\in\mbb F_q}\left\{(x,cx):x\neq 0,x^2\in\frac{1}{\eta_1+\eta_2c+\eta_3c^2}H\right\}\cup\left\{(0,y):y\neq 0,y^2\in\frac{1}{\eta_3}H\right\}.
    \]
\end{lem}

Next we turn to $\mc S$. We already noted in Section \ref{sect:mosaics} that the parallel classes of $D$ are in one-to-one correspondence with those of $AG(2,q)$, i.e., with the slopes from $\mbb F_q\cup\{\infty\}$. 

For the description of the elements of a parallel class, we need the \textit{(absolute) trace} of an element $x$ of $\mbb F_q$ defined by
\[
    \Tr(x)=x+x^2+\cdots+x^{2^{t-1}}.
\]
The trace is an $\mbb F_2$-linear form from $\mbb F_q$ onto $\mbb F_2$. Every linear form $\xi$ from $\mbb F_q$ to $\mbb F_2$ corresponds to a unique element $\beta\in\mbb F_q$ such that $\xi(x)=\Tr(\beta x)$ for all $x\in\mbb F_q$ (see \cite[Theorem 2.23]{LN_finite_fields}). We denote by $H^\perp$ the $(t-\ell)$-dimensional subspace of $\mbb F_q$ consisting of those elements whose corresponding linear form vanishes  on $H$.

We will also use the following facts on polynomials. The first one is \cite[Theorem 2.25]{LN_finite_fields}, the second one is elementary.

\begin{fact}\label{fact:poly_fact}
    \begin{enumerate}
        \item The polynomial $x^2+x+\alpha$, with $\alpha\in\mbb F_q$, has a root in $\mbb F_q$ if and only if $\Tr(\alpha)=0$.
        \item Let $F(x)=\alpha x^2+\beta x+\gamma$ be a polynomial over $\mbb F_q$. Then $F(c)=0$ if and only if $\alpha c/\beta$ is a root of
        \[
            x^2+x+\frac{\alpha\gamma}{\beta^2}.
        \]
    \end{enumerate}
\end{fact}

We have the following lemma.

\begin{lem}\label{lem:Denniston_S}
    For any $c\in\mbb F_q\cup\{\infty\}$, denote by $\mc U_c$ the set of those $d\in\mbb F_q$ for which $L_{c,d}$ has nonempty intersection with $\mc X$. If $c\in\mbb F_q$, then 
    \begin{equation}\label{eq:U_c_form}
        \mc U_c=\left\{d\neq 0:d^{-2}\notin\frac{\eta_2^2}{\eta_1+\eta_2c+\eta_3c^2}H^\perp\right\}\cup\{0\}.
    \end{equation}
    If $c=\infty$, then
    \[
        \mc U_c=\left\{d\neq 0:d^{-2}\notin\frac{\eta_2^2}{\eta_3}H^\perp\right\}\cup\{0\}.
    \]
\end{lem}

\begin{proof}
	We use Fact \ref{fact:poly_fact}. Let $c,d\in\mbb F_q$. For $L_{c,0}$ we already know that it has nonempty intersection with $\mc X$, so assume $d\neq 0$. Then $L_{c,d}$ has nonempty intersection with $\mc X$ if and only if the polynomial
    \[
        F(x)=(\eta_1+\eta_2c+\eta_3c^2)x^2+\eta_2dx+\eta_3d^2
    \]
    assumes a value in $H$ for some $x\in\mbb F_q$. By Fact \ref{fact:poly_fact}, this is the case if and only if there exists a $z\in H$ such that
    \[
        \Tr\left(\frac{(\eta_1+\eta_2 c+\eta_3c^2)(\eta_3 d^2+z)}{\eta_2^2d^2}\right)=0.
    \]
    The term inside the trace can be written as
    \[
        \frac{(\eta_1+\eta_2c+\eta_3c^2)z}{\eta_2^2d^2}+\frac{\eta_1\eta_3}{\eta_2^2}+\left(\frac{\eta_3 c}{\eta_2}+\frac{\eta_3^2 c^2}{\eta_2^2}\right).
    \]
    The sum inside the large brackets has trace zero since $\Tr(\alpha)+\Tr(\alpha^2)=0$ for all $\alpha\in\mbb F_q$. The trace of $(\eta_1\eta_3)/\eta_2^2$ equals 1 due to the irreducibility of $Q$. It follows that $z\in H$ satisfies $F(x)=z$ for some $x\in\mbb F_q$ if and only if 
    \[
        \Tr\left( \frac{(\eta_1+\eta_2c+\eta_3c^2)z}{\eta_2^2d^2} \right)=1.
    \]
	Hence a nonzero $d\in\mbb F_q$ is \textit{not} contained in $\mc U_c$ if and only if 
    \[
        \frac{(\eta_1+\eta_2c+\eta_3c^2)}{\eta_2^2d^2}\in H^\perp,
    \]
    which immediately shows \eqref{eq:U_c_form}. The proof for $c=\infty$ is analogous.
\end{proof}

\subsection{Property (D1)}

$D$ is explicit if it satisfies properties (D1) and (D2) formulated in Section \ref{subsect:gen_constr}. Here we show that it satisfies (D1) for suitable $H$.  Let $\Theta=\{1,\vartheta,\vartheta^2,\ldots,\vartheta^{t-1}\}$ be a polynomial basis of $\mbb F_q$. We take $H$ as the span of $1,\ldots,\vartheta^{\ell-1}$. Let $\Phi_H:[k]\to\mbb F_2^\ell$ be a bijection which in time $\poly(\log k)$ associates to every number from $[k]$ a unique element of $H$, represented in terms of $\Theta$, such that $\Phi_H(0)=0$.

Denote by $\Phi_{\mc R}:[q+1]\to\mbb F_2^t\cup\{\infty\}$ a $\poly(\log q)$ time  bijection between $[q+1]$ and the set of slopes $\mc R=\mbb F_q\cup\{\infty\}$, where $\Phi(\tilde\imath)$ for any $\tilde\imath\in[q]$ is the representation in the basis $\Theta$ of a unique element of $\mbb F_q$. 

Arithmetic operations in $\mbb F_q$ can be performed efficiently in $\Theta$, as well as the computation of the square root \cite[Corollary 7.1.2]{BachShallit_AlgNT}. Hence using $\Phi_{\mc R}$ and $\Phi_H$, one obtains a mapping $\Phi_{\mc X}:[v]\to\mbb F_2^t$ which to every element of $[v]$ associates the $\Theta$-representation of a unique element of $\mc X$ (see Lemma \ref{lem:Denniston_X}). This mapping is computable in time $\poly(\log v)$. 

To the basis $\Theta$ there exists a \textit{dual} basis $Z=\{\zeta_1,\ldots,\zeta_t\}$ satisfying
\[
    \Tr(\zeta_i\vartheta^j)=\delta_{ij}.
\]
$H^\perp$ is the span of $\{\zeta_{\ell},\ldots,\zeta_{t-1}\}$. Denote by $T$ the change-of-basis matrix representing every $\zeta_i$ in terms of $\Theta$. Then for any $c$, there exists a bijective mapping $\Phi_{\mc U_c}:[a]\to\mbb F_2^t$ which to any element of $[a]$ first associates the $Z$-representation of an element of $(\mbb F_q\setminus H^\perp)\cup\{0\}$, then changes the basis to $\Theta$ using $T$, and finally does the necessary arithmetic to obtain an element of $\mc U_c$. The values of this mapping can be computed in time $\poly(t)=\poly(\log a)$.

Now assume we are given numbers $\tilde x\in[v]$ and $\tilde\imath\in[q+1]$, corresponding to the point $(x,\tilde cx)\in\mc X$ and the parallel class $c\in\mbb F_q\cup\{\infty\}$ via $\Phi_{\mc X}$ and $\Phi_{\mc R}$. We want to find the intercept $d$ such that $(x,\tilde cx)\in L_{c,d}$. If $c\in\mbb F_q$, then $d=(c+\tilde c)x$. If $c=\infty$, then $d=x$. It is straightforward to do these computations in $\Theta$. The result is transformed to a number from $[a]$ via $\Phi_{\mc U_c}^{-1}$. The representation of $d$ in $[a]$ can be found from inputs $\tilde x$ and $\tilde\imath$ in $\poly(\log v)$ time.

\subsection{Property (D2)}

Let $(c,d)\in\mc S$ be given. We want to find the set $B_{c,d}$ of those elements of $\mc X$ which are incident with $(c,d)$ in $D$. For $d=0$, we have $L_{c,d}=\mc X_c\cup\{0\}$. Now we consider the case $d\neq 0$. Let 
\[
    \mc R_{c,d}=\{\tilde c\in\mc R:L_{c,d}\cap\mc X_{\tilde c}\neq\emptyset\}.
\]
Once we know the set $\mc R_{c,d}$, we can for every $\tilde c\in\mc R_{c,d}$ find the unique point at the intersection of $L_{c,d}$ and $\mc X_{\tilde c}$. If $\tilde c\in\mbb F_q$, this point has the form $(x,\tilde cx)$ for
\[
    x=\frac{d}{c+\tilde c}
\]
(clearly, $c\neq\tilde c$). If $\tilde c=\infty$, the point at the intersection of $L_{c,d}$ and $\mc X_\infty$ is given by $(0,d)$.

For $c\in\mbb F_q$, define the set
\[
    H_{c,d}=\left\{z\in H:\Tr\left(\frac{(\eta_1+\eta_2c+\eta_3c^2)z}{\eta_2^2d^2}\right)=1\right\}
\]
and, for every $z\in H_{c,d}$, the polynomial 
\[
    G_{c,d,z}(w)=w^2+w+\frac{(\eta_1d^2+c^2z)(z+\eta_3d^2)}{\eta_2^2d^4}.
\]
For $c=\infty$, we set 
\[
    H_{c,d}=\left\{z\in H:\Tr\left(\frac{\eta_3z}{\eta_2^2d^2}\right)=1\right\}
\]
and define, for all $z\in H_{c,d}$, the polynomial
\[
    G_{c,d,z}(w)=w^2+w+\frac{\eta_3d^2(\eta_1d^2+z)}{\eta_2^2d^4}.
\]
All $H_{c,d}$ are nonempty due to the proof of Lemma \ref{lem:Denniston_S}.

\begin{lem}\label{lem:R_c_d}
    If $c\in\mbb F_q$ and $\eta_3d^2\notin H$, then 
    \[
        \mc R_{c,d}=\left\{\frac{\eta_2d^2w}{z+\eta_3d^2}:w\text{ root of }G_{c,d,z},\;z\in H_{c,d}\right\}.
    \]
    If $c\in\mbb F_q$ and $\eta_3d^2\in H$, then
    \[
        \mc R_{c,d}=\left\{\frac{\eta_2d^2w}{z+\eta_3d^2}:w\text{ root of }G_{c,d,z},\;z\in H_{c,d}\setminus\{\eta_3d^2\}\right\}\cup\left\{\frac{\eta_1+\eta_3c^2}{\eta_2},\infty\right\}.
    \]
    If $c=\infty$, then 
    \[
        \mc R_{c,d}=\left\{\frac{\eta_2w}{\eta_3}:w\text{ root of }G_{c,d,z},\;z\in H_{c,d}\right\}.
    \]
\end{lem}

\begin{proof}
We start with the case $c\in\mbb F_q$ and $\eta_3d^2\notin H$. There exists an $x\in\mbb F_q$ such that $(x,\tilde cx)\in\mc X_{\tilde c}\cap L_{c,d}$ if and only if 
\[
    \frac{d^2}{c^2+\tilde c^2}\in\frac{1}{\eta_1+\eta_2\tilde c+\eta_3\tilde c^2}H,
\]
which is equivalent to the existence of a $z\in H$ such that
\begin{equation}\label{eq:H_c_d_quadr_eq}
    (z+\eta_3d^2)\tilde c^2+\eta_2d^2\tilde c+\eta_1d^2+c^2z=0.
\end{equation}
By Fact \ref{fact:poly_fact}.2), $\tilde c$ is a root of this equation if and only if $\eta_2^{-1}d^{-2}(z+\eta_3d^2)\tilde c$ is a root of $G_{c,d,z}$. 

It follows from the proof of Lemma \ref{lem:Denniston_S} that the set of $z\in H$ for which $G_{c,d,z}$ has a root in $\mbb F_q$ necessarily is equal to $H_{c,d}$. One can also check this directly using Fact \ref{fact:poly_fact}. Write the constant term of $G_{c,d,z}$ as
\[
    \frac{\eta_1\eta_3}{\eta_2^2}+\frac{(\eta_1+\eta_2c+\eta_3c^2)z}{\eta_2^2d^2}+\left(\frac{cz}{\eta_2d^2}+\frac{c^2z^2}{\eta_2^2d^4}\right).
\]
As in the proof of Lemma \ref{lem:Denniston_S}, one concludes that the set of $z\in H$ where $G_{c,d,z}$ has a root in $\mbb F_q$ is given by $H_{c,d}$, as claimed. Each root $w$ of $G_{c,d,z}$ gives a root $\tilde c$ of \eqref{eq:H_c_d_quadr_eq}, and this gives the claimed form of $\mc R_{c,d}$. 

Now assume $c\in\mbb F_q$ and $\eta_3d^2\in H$. Then \eqref{eq:H_c_d_quadr_eq} has two distinct roots as in the previous case unless $z=\eta_3d^2$, in which case the quadratic term vanishes. This gives $\tilde c=(\eta_1+\eta_3c^2)/\eta_2$ (the irreducibility of $Q$ ensures $\eta_2\neq 0$). One checks directly that $\infty\in\mc R_{c,d}$.

The case $c=\infty$ is treated analogously to the first case.
\end{proof}

\begin{rem}
	We note that for distinct $z,z'\in H_{c,d}$, the roots of the corresponding $G_{c,d,z}$ and $G_{c,d,z'}$ are different. This follows from a simple counting argument. Assume $c\in\mbb F_q$ and $\eta_3d^2\notin H$, the other cases are analogous. Since $\lvert\mc R_{c,d}\rvert=\lvert B_{c,d}\rvert=2^\ell$, we know from the proof of Lemma \ref{lem:R_c_d} that the total number of roots of $G_{c,d,z}$ as $z$ ranges over $H_{c,d}$ is $2^\ell$. Now $G_{c,d,z}$ has two distinct roots in $\mbb F_q$ for every $z\in H_{c,d}$, for if $G_{c,d,z}(w)=0$, then $G_{c,d,z}(w+1)=0$. Moreover, $H_{c,d}$ is the coset of an $(\ell-1)$-dimensional subspace of $H$.
\end{rem}

It remains to check that (D2) is satisfied. Let $\tilde\imath\in[q+1]$ correspond to a parallel class and $\tilde\kappa\in[a]$ to an element of this parallel class. Through the mapping $\Phi_{\mc R}$, one associates to $\tilde\imath$ a slope $c\in\mbb F_q\cup\{\infty\}$. Then $\Phi_{\mc U_c}(\tilde\kappa)$ gives an intercept $d$ such that $(c,d)\in\mc S$, where both $c$ and $d$ are represented in the basis $\Theta$. It remains to show that the set $\mc R_{c,d}$ can be enumerated in polylogarithmic time. The first task is to find $H_{c,d}$. We shall use that if $\beta=\sum_{i=0}^{t-1}\beta_i\zeta_i$ and $z=\sum_{i=0}^{t-1}z_i\vartheta^i$, then $\Tr(\beta z)=\sum_{i=0}^{t-1}\beta_iz_i$.

Assume that $c\in\mbb F_q$ and $\eta_3d^2\notin H$ (which can be checked by representing $\eta_3d^2$ in the basis $Z$). The other cases are similar. Using the $\Theta$-representations of $c$ and $d$, compute
\[
    \beta=\frac{\eta_1+\eta_2c+\eta_3c^2}{\eta_2^2d^2}.
\]
Transform the result to $Z$. Now assume that $\beta=\sum_{i=0}^{t-1}\beta_i\zeta_i$. Since $H_{c,d}$ is nonempty, the linear form $z\mapsto\Tr(\beta z)$ does not vanish on $H$. Hence $\beta_i=1$ for some $0\leq i\leq\ell-1$, say $\beta_{\ell-1}=1$. One can now enumerate the $\Theta$-representations of all elements of $H_{c,d}$ by enumerating all sequences $z_0,\ldots,z_{\ell-2}$ and choosing $z_{\ell-1}$ such that $\sum_{i=0}^{\ell-1}\beta_iz_i=1$. 

Given $z\in H_{c,d}$, it remains to find both roots of $G_{c,d,z}$. This means that one has to solve the in homogeneous linear equation
\[
    \tilde c^2+\tilde c=\frac{(\eta_1d^2+z)(\eta_1d^2+c^2z)}{\eta_2^2d^2}.
\]
If $\tilde c$ satisfies this equation, then $\tilde c+1$ is the other solution. This equation can be solved in polylogarithmic time in $q$, and so it is possible to find the points of $\mc X$ incident with $(c,d)$ in polylogarithmic time.

\begin{rem}
    In the way they were described here, it appears obvious that property (D1) requires less computation than (D2) for Denniston's BIBD, although both operations have the same complexity class. For the computation of the functional form, only the former operation is necessary. 
    
    For the computation of the randomized inverse of the functional form, recall that it is possible, as pointed out in Section \ref{subsect:complexity}, to first choose the point from $\mc X$ uniformly at random, and then to choose an $s\in\mc S$ such that $f(x,s)=\alpha$ if $\alpha$ is the message to be transmitted. In this approach, it is sufficient to randomly choose $c\in\mc R$ and then to solve for the intercept $d\in\mc A$ as in (D1).
\end{rem}

\section*{Acknowledgments}

The authors would like to thank Eike Kiltz for discussions about the achievability of semantic security using modular coding schemes. H.~Boche would also like to thank Marc Geitz, Oliver Holschke and Frank Fitzek for discussions about the application of modular wiretap coding schemes in communication networks.

Both authors were supported by the Deutsche Forschungsgemeinschaft (DFG, German Research Foundation) under Germany's Excellence Strategy - EXC 2092 CASA - 390781972. H.~Boche was also partly supported by the National Research Initiative of the German Ministry for Education and Research (BMBF) on 6G Communication Systems through the research hub 6G-life (16KISK002).

\end{document}